\documentclass[reqno,oneside,12pt]{amsart}

\usepackage{amssymb,graphicx,braket,amsmath,amsfonts,amsthm,bm,enumitem}
\usepackage[mathscr]{euscript}

\usepackage{hyperref}
\usepackage[utf8]{inputenc}
\usepackage{color}
\usepackage{amsaddr}
\usepackage[T1]{fontenc}
\usepackage[utf8]{inputenc}
\usepackage{tikz}
\usetikzlibrary{shapes,arrows,er,positioning}
\usetikzlibrary{shapes.multipart}
\usepackage{stmaryrd}
\usepackage{fancyhdr}
\usepackage{empheq}
\usepackage{xcolor}
\usepackage{array}
\usepackage{wasysym}
\usetikzlibrary{arrows}
\usetikzlibrary{arrows.meta}
\allowdisplaybreaks

\numberwithin{equation}{section}

\theoremstyle{plain}
\newtheorem{thm}{Theorem}[section]
\newtheorem{cor}[thm]{Corollary}
\newtheorem{lem}[thm]{Lemma}
\newtheorem{prop}[thm]{Proposition}
\newtheorem{conj}[thm]{Conjecture}

\theoremstyle{definition}
\newtheorem{defin}[thm]{Definition}

\newtheorem{rem}[thm]{Remark}

\DeclareMathOperator*{\Res}{Res}

\numberwithin{equation}{section}

\def\mainmatter{\def\baselinestretch{1}\normalfont \setlength{\parskip}{0.5em}}

\usepackage[margin=1in]{geometry}

\newcommand{\bC}{\mathbb{C}}
\newcommand{\bZ}{\mathbb{Z}}
\newcommand{\bP}{\mathbb{P}}

\newcommand{\cS}{\mathcal{S}}
\newcommand{\cP}{\mathcal{P}}
\newcommand{\cO}{\mathcal{O}}

\newcommand{\cA}{\mathcal{A}}
\newcommand{\cB}{\mathcal{B}}
\newcommand{\cR}{\mathcal{R}}
\newcommand{\cW}{\mathcal{W}}

\newcommand{\sQ}{\mathscr{Q}}

\newcommand{\sA}{\mathscr{A}}

\makeatletter
\renewcommand{\section}{\@startsection
{section}
{1}
{\z@}
{-\baselineskip}
{0.8\baselineskip}
{\centering\scshape\large}} 

\renewcommand{\subsection}{\@startsection
{subsection}
{2}
{\z@}
{-0.8\baselineskip}
{0.5\baselineskip}
{\normalfont \bf \normalsize}} 

\renewcommand{\subsubsection}{\@startsection
{subsubsection}
{3}
{\z@}
{-0.8\baselineskip}
{0.5\baselineskip}
{\normalfont \bf \normalsize}} 
\makeatother

\usepackage[backend=bibtex,style=alphabetic,maxbibnames=10,sorting=nyt,giveninits]{biblatex}
\renewbibmacro{in:}{}

\begin{document}
\title{Deformation and quantisation condition \\
of the $\sQ$-top recursion}

\author{Kento Osuga}
\address{Graduate School of Mathematical Sciences, University of Tokyo,\\
3-8-1 Komaba, Meguro, Tokyo, 153-8914, Japan}
\email{osuga@ms.u-tokyo.ac.jp}

\begin{abstract}
We consider a deformation of a family of hyperelliptic refined spectral curves and investigate how deformation effects appear in the hyperelliptic refined topological recursion as well as the $\sQ$-top recursion. We then show a coincidence between a deformation condition and a quantisation condition in terms of the $\sQ$-top recursion on a degenerate elliptic curve. We also discuss a relation to the corresponding Nekrasov-Shatashivili effective twisted superpotential.
\end{abstract}

\newpage\maketitle
\setcounter{tocdepth}{1}\tableofcontents\mainmatter

\newpage
\section{Introduction}
The purpose of the present paper is twofold. One is to describe the so-called variational formula in the framework of the hyperelliptic refined topological recursion as well as the $\sQ$-top recursion proposed in \cite{KO22,O23}. The other is to reveal an intriguing coincidence between a deformation condition and a quantisation condition in terms of the $\sQ$-top recursion as an application of the variational formula. 

\subsection{Motivations and Backgrounds}

Since motivations and backgrounds of a refinement of topological recursion are discussed in \cite{KO22,O23} in detail, we only give a brief review of recent developments in this direction.

 As defined in \cite{O23} (and in \cite{KO22} for a special class of genus-zero curves), a \emph{hyperelliptic refined spectral curve} $\cS_{\bm\kappa,\bm\mu}$ consists of three data: a compactified and normalised Torelli-marked hyperelliptic curve $C=(\Sigma,x,y)$ of genus $\tilde g$\footnote{We abuse the terminology and include curves of $\tilde g=0,1$.}, complex parameters $\bm\kappa$ associated with the Torelli markings, and complex parameters $\bm\mu$ associated with non-ramified zeroes and poles of a differential $ydx$\footnote{Strictly speaking, $ydx$ has to be anti-symmetrised in terms of the hyperelliptic involution $\sigma$.}. We often drop `hyperelliptic' for brevity. Taking a refined spectral curve as initial data, the \emph{refined topological recursion} constructs an infinite sequence of multidifferentials $\omega_{g,n}$ on $\Sigma^n$ labeled by $n\in\bZ_{\geq0}$ and $g\in\frac12\bZ_{\geq0}$ --- $g$ is different from the genus of $\Sigma$. \cite{KO22,O23} proved or conjectured properties of $\omega_{g,n}$. Several results based on matrix models have also been discussed in e.g. \cite{CE06,BMS10,C10,CEM10,MS15}.
 
 The multidifferentials $\omega_{g,n}$ polynomially depend on the refinement parameter $\sQ$, up to $\sQ^{2g}$. It is easy to see that the $\sQ$-independent part precisely corresponds to the Chekhov-Eynard-Orantin topological recursion \cite{CE05,CEO06,EO07}. As shown in \cite{O23}, it turns out that the $\sQ$-top degree part also give rise to a self-closed recursion, and we call it the \emph{$\sQ$-top recursion}. That is, the Chekhov-Eynard-Orantin topological recursion and the $\sQ$-top recursion are a subsector of the full refined topological recursion, and we respectively denote differentials in each subsector by $\omega_{g,n}^{{\rm CEO}}$ and $\varpi_{g,n}$ to notationally distinguish from $\omega_{g,n}$.

For a family of hyperelliptic curves $C(\bm t)$ with some complex parameters $\bm t$, one can consider the corresponding family of refined spectral curves $\cS_{\bm\kappa,\bm\mu}(\bm t)$ (with mild restrictions, e.g. ramification points should not collide each other under deformation of parameters). As a consequence, $\omega_{g,n}$ also depend on the parameters $\bm t$, and one may ask: how do $\omega_{g,n}$ vary under a deformation with respect to $\bm t$? 

In the unrefined setting, this point has already been addressed in \cite{EO07,E17}, and we know how $\omega_{g,n}^{{\rm CEO}}(\bm t)$ varies which is known as the \emph{variational formula}\footnote{The variational formula in the unrefined setting is not limited to hyperelliptic curves.}. It can be thought of as a generalisation of the Seiberg-Witten relation \cite{SW94-1,SW94-2}. However, it turns out that there is a subtlety and difficulty when one tries to apply the original Eynard-Orantin proof to the refined setting. Thus, we provide an equivalent interpretation of the variational formula (Definition \ref{def:var}) which becomes easier to apply to the refined topological recursion. With this perspective, we are able to state a refined analogue of the variational formula.

\subsection{Summary of main results}
The first achievement of the present paper is to prove the variational formula for the refined topological recursion, when $\Sigma=\bP^1$ (Theorem \ref{thm:var}). However, since we have to fix several notations and technical aspects in order to remove the subtlety mentioned above, it is hard to state the variational formula here and we leave all the details to Section \ref{sec:Var}. Roughly speaking, it states that a certain deformation $\delta_{t}*\omega_{g,n}$ with respect to $t\in\bm t$ is related to an integral of $\omega_{g,n+1}$ as follows:
\begin{equation}
    \delta_{t}*\omega_{g,n}=\int_{p\in\gamma}\Lambda(p)\cdot\omega_{g,n+1},\label{intro:var}
\end{equation}
where $(\gamma,\Lambda)$ is defined in Definition \ref{def:pair}. Let us emphasise that, in contrast to the unrefined setting, the variational formula \eqref{intro:var} holds only when a refined spectral curve $\cS_{\bm\kappa,\bm\mu}(\bm t)$ satisfies a certain condition which we call the \emph{refined deformation condition} (Definition \ref{def:condi}). See Section \ref{sec:Var} for more details. Note that some properties of the refined topological recursion are still conjectural when $\Sigma\neq\bP^1$ \cite{O23}, hence the variational formula also remains conjectural in this case. We also note that \cite{CEM10} discuss a similar formula in a different refined setting.

Another achievement of the present paper is to uncover an intriguing coincidence between the refined deformation condition and what we call the $\sQ$-top quantisation condition defined as follows. It is shown in \cite{O23} that the $\sQ$-top recursion naturally constructs a second-order ordinary differential operator, called the \emph{$\sQ$-top quantum curve}. For a refined spectral curve $\cS_{\bm\kappa,\bm\mu}(\bm t)$  whose underlying curve is given by $y^2=Q_0(x)$, the associated $\sQ$-top quantum curve is written in the  following form
\begin{equation}
      \left( \epsilon_1^2\frac{d^2}{dx^2}-Q_0(x)-\sum_{k\in\bZ_{\geq1}}\epsilon_1^{k}\cdot Q_k(x)\right) \psi^{\sQ{\rm -top}}(x)=0,\label{intro:QC0}
     \end{equation}
where $\epsilon_1$ is a formal parameter, $Q_k(x)$ is a rational function of $x$ determined by $\{\varpi_h\}_h$ for $2h<k$, and the logarithmic derivative of $ \psi^{\sQ{\rm -top}}(x)$ is a formal sum of $\epsilon_1^{2g-1}\cdot\varpi_{g,1}$ over $g$.  In the context of topological recursion, one may sometime require a condition on quantisation that the set of poles of $Q_k(x)$ should be a subset of poles of $Q_0(x)$. Therefore, we say that a refined spectral curve $\cS_{\bm\kappa,\bm\mu}(\bm t)$ satisfies the \emph{$\sQ$-top quantisation condition}, if the $\sQ$-top quantum curve respects the pole structure of $Q_0(x)$ (Definition \ref{def:QQC}) --- existence of a quantum curve in the full refined setting is proven only for a special class of genus-zero curves \cite{KO22} and in this case one can analogously consider the \emph{refined quantisation condition}.

In order to deliver a clear picture about the coincidence between the refined deformation condition and the $\sQ$-top quantisation condition, let us focus on the following example. For $t\in\bC^*$, we consider a one-parameter family of curves $C_t=(\bP^1,x,y)$ where meromorphic functions $(x,y)$ satisfy:
\begin{equation}
    y^2-Q_0(x;t)=0,\qquad Q_0(x;t)=4\left(x-q_0\right)^2\cdot\left(x+2q_0\right),\qquad q_0=\sqrt{-\frac{t}{6}}.\label{singP1}
\end{equation}
This is the curve associated with the zero-parameter solution of the Painlev\'{e} I equation, and $t$ plays the role of the Painlev\'{e} time \cite{IS15}. Since $ydx$ has a simple zero at the preimages of $x=q_0$, the corresponding refined spectral curve $\cS_{\mu}(t)$ carries one parameter $\mu\in\bC$, and $\omega_{g,n}$ depend both on $t$ and $\mu$.

In this example, it turns out that $\cS_{\mu}(t)$ satisfies the refined deformation condition if and only if $\mu$ is set to a special value $\mu=\mu_0$. (Proposition \ref{prop:rdc2}). On the other hand, one can show that $Q_{k\geq2}(x;t,\mu)$ has a pole at $x=q_0$ for a generic $\mu$, which is a \emph{zero} of $Q_0(x;t)$. However, it turns out that when $\mu=\mu_0$, such poles disappear for all $k$, and thus, the $\sQ$-top quantisation condition is satisfied (Proposition \ref{prop:QQC}).  Therefore, we observe that \emph{the refined deformation condition and the $\sQ$-top quantisation condition precisely agree}, even though they originated from two different requirements. It is interesting to see whether this coincidence holds in other curves, e.g. curves discussed in \cite{IMS16} in relation to other Painlev\'{e} equations.

When $\mu=\mu_0$, the variational formula gives a relation between $Q_k(x;t,\mu_0)$ in \eqref{intro:QC0} and a derivative of $F^{\sQ\text{-top}}_g:=\varpi_{g,0}$ with respect to $t$ --- the former appears in the $\sQ$-top quantisation and the latter is a consequence of a deformation of a refined spectral curve:
 \begin{thm}[Theorem \ref{thm:QC}]\label{intro:thm}
     Consider the above family of refined spectral curves $\cS_{\mu_0}(t)$ satisfying the refined deformation condition and also the $\sQ$-top quantisation condition. Then, the associated $\sQ$-top quantum curve is given in the following form:
     \begin{equation}
        \left( \epsilon_1^2\frac{d^2}{dx^2}-4x^3-2tx-2\sum_{g\in\frac12\bZ_{\geq0}}\epsilon_1^{2g} \frac{\partial F^{\sQ{\rm -top}}_g}{\partial t}\right)\psi^{\sQ{\rm -top}}(x)=0.\label{intro:QC}
     \end{equation}
 \end{thm}

It is crucial to remark that there is no $\epsilon_1^2\partial/\partial t$ term in \eqref{intro:QC}, in contrast to the quantum curve derived in \cite{IS15,Iwa19} within the framework of the Chekhov-Eynard-Orantin topological recursion. Instead, a similar differential operator to \eqref{intro:QC} has appeared in the context of conformal blocks in the semi-classical limit, or the so-called \emph{Nekrasov-Shatashivili limit} e.g. \cite{LN21,LN22,BIPT22}. Note that they consider a genus-one curve whose singular limit becomes \eqref{singP1}, and we expect that the form of \eqref{intro:QC} remains the same for the corresponding genus-one curve. Importantly, their arguments and Theorem \ref{intro:thm} suggest a conjectural statement that $F^{\sQ{\rm -top}}_g$ agrees with the so-called \emph{Nekrasov-Shatashivili effective twisted superpotential} $\cW_g^{{\rm eff}}$ \cite{NS09}, when a refined spectral curve is chosen appropriately:
\begin{equation}
    \sum_{g\in\frac12\bZ_{\geq0}}\epsilon_1^{2g}F^{\sQ{\rm -top}}_g \overset{?}{=} \sum_{g\in\frac12\bZ_{\geq0}}\epsilon_1^{2g}\cW_g^{{\rm eff}}:=\epsilon_1\epsilon_2 \log Z^{\rm Nek}\big|_{\epsilon_2=0},
\end{equation}
where $Z^{\rm Nek}$ is the corresponding Nekrasov partition function \cite{Nek03} and the equality should be considered as a formal series in $\epsilon_1$. See e.g. \cite{NRS11,HK17,HRS21} for more about Nekrasov-Shatashivili effective twisted superpotentials. Note that for the curve associated with the Painlev\'{e} I equation, the Nekrasov partition function is not defined from an irregular conformal block perspective, whereas $F^{\sQ{\rm -top}}_g$ is perfectly well-defined. We hope that the present paper together with the notion of the $\sQ$-top recursion \cite{O23} sheds light on verifying the above statement and also triggers a new direction between topological recursion, the $\sQ$-top recursion, and invariants in the Nekrasov-Shatashivili limit (e.g. a role of $\varpi_{g,n\geq2}$).

\subsubsection*{Acknowledgement}

The author thanks Nitin Chidambaram, Elba Garcia-Failde,  Hajime Nagoya, Lotte Hollands, Kohei Iwaki, Omar Kidwai, Oleg Lisovyy, Nicolas Orantin for discussions and correspondences. This work is supported by JSPS KAKENHI Grant Number 22J00102, 22KJ0715, and 23K12968, also in part by KAKENHI Grant Number 20K14323 and 21H04994.

\newpage

\section{Definitions}\label{sec:Def}
We briefly review the refined topological recursion proposed in \cite{KO22,O23}. We refer to the readers \cite[Section 2]{O23} for more details. 

\begin{defin}[\cite{KO22,O23}]\label{def:RSC}
A \emph{hyperelliptic refined spectral curve} $\cS_{{\bm \mu}, {\bm \kappa}}$ consists of the collection of the following data:
\begin{itemize}
    \item $(\Sigma,x,y)$: a connected compact Riemann surface of genus $\tilde g$ with two meromorphic functions $(x,y)$ satisfying
    \begin{equation}
        y^2-Q_0(x)=0,
    \end{equation}
    where $Q_0(x)$ is a rational function of $x$ which is not a complete square. We denote by $\sigma:\Sigma\to\Sigma$ the hyperelliptic involution of $x:\Sigma\to\bP^1$ and by $\cR$ the set of ramification points of $x$, i.e. set of $\sigma$-fixed points.
    \item $(\cA_i,\cB_i,\kappa_i)$: a choice of a canonical basis $\cA_i,\cB_i\in H_1(\Sigma,\bZ)$ and associated parameters $\kappa_i\in\bC$ for $i\in\{1,..,\tilde g\}$,
    \item $(\widetilde\cP_+,\mu_p)$: a choice of a decomposition $\widetilde\cP_+\sqcup\sigma(\widetilde{\cP}_+)=\widetilde\cP$ and associated parameters $\mu_p\in\bC$ for all $p\in\widetilde\cP_+$ where $\widetilde\cP$ is the set of unramified zeroes and poles of $ydx$.
\end{itemize}
\end{defin}

Let us fix some notation before defining the refined topological recursion. First of all, throughout the present paper, $g,h$ are in $\frac12\bZ_{\geq0}$, $n,m$ in $\bZ_{\geq0}$, $i,j$ in $\{1,..,\tilde g\}$ and $a,b$ in $\{0,..,n\}$. We denote by $B$ the fundamental bidifferential of the second kind, and for a choice of representatives $\sA_i$ of $\cA_i$ for each $i$, we denote by $\eta_\sA^p$ the fundamental differential of the third kind for $p\in\Sigma$ normalised along each $\sA_i$-cycle. We write $p_a\in\Sigma$ for each $a$, $J:=(p_1,..,p_n)\in(\Sigma)^{n}$, and $J_0:=\{p_0\}\cup J\in(\Sigma)^{n+1}$. Assuming $p_a\not\in\cR\cup\sigma(\cP_+)$ for all $a$, we denote by $C_+$ a connected and simply-connected closed contour such that it contains all points in $J_0\cup\cP_+$ and no points in $\cR\cup\sigma(J_0\cup\cP_+)$. With the assumption on $p_a$, one can always find such a contour and we drop the $n$-dependence on $C_+$ for brevity. Similarly, we denote by $C_-$ a connected and simply-connected closed contour containing all points in $\cR\cup\sigma(J_0\cup\cP_+)$ but not points in $J_0\cup\cP_+$. We call $p\in\cR$ \emph{ineffective} if $ydx$ is singular at $p$, and \emph{effective} otherwise. We denote by $\cR^*$ the set of effective ramification points. We denote by $\cP^{0,\infty}_+\cup\sigma(\cP_+^{0,\infty})$ the set of unramified zeroes and poles of $ydx$ respectively, and denote by $C_-^{\mathfrak p}$ a connected and simply-connected closed contour inside $C_-$ but not containing points in $\sigma(\cP_+^{\infty}
)$. Finally, we fix $\sQ\in\bC$ and we call it the refinement parameter.

\begin{defin}[\cite{KO22,O23}]
    Given a hyperelliptic refined spectral curve $\cS_{{\bm \mu}, {\bm \kappa}}$, the \emph{hyperelliptic refined topological recursion} is a recursive definition of multidifferentials $\omega_{g,n+1}$ on $(\Sigma)^{n+1}$ by the following formulae:
    \begin{align}
        \omega_{0,1}(p_0):&=y(p_0)\cdot dx(p_0),\label{w01}\\
    \omega_{0,2}(p_0,p_1):&=-B(p_0,\sigma(p_1)),\label{w02}\\
    \omega_{\frac12,1}(p_0):&=\frac{\sQ}{2}\left(-\frac{d\Delta y(p_0)}{\Delta y(p_0)}+\sum_{p\in\widetilde{\cP}_+}\mu_p\cdot\eta^p_\sA(p_0)+\sum_{i=1}^{\tilde g}\kappa_i\cdot \int_{\cB_i}B(\cdot,p_0)\right),\label{w1/2,1}
    \end{align}
    and for $2g-2+n\geq0$,
      \begin{equation}
            \omega_{g,n+1}(J_0):=\frac{1}{2 \pi i}\left(\oint_{p\in C_+}-\oint_{p\in C_-}\right)\frac{\eta_\sA^p(p_0)}{4\omega_{0,1}(p)}\cdot {\rm Rec}_{g,n+1}^{\sQ}(p,J),\label{RTR}
        \end{equation}
        where 
        \begin{align}
           {\rm Rec}_{g,n+1}^{\sQ}(p_0;J):=&\sum^*_{\substack{g_1+g_2=g\\J_1\sqcup J_2=J}} \omega_{g_1,n_1+1}(p_0,J_1)\cdot \omega_{g_2,n_2+1}(p_0,J_2)+\sum_{t\sqcup I=J}\frac{dx(p_0)\cdot dx(t)}{(x(p_0)-x(t))^2}\cdot \omega_{g,n}(p_0,I)\nonumber\\
    &+ \omega_{g-1,n+2}(p_0,p_0,J)+\sQ \cdot dx \cdot d_0\frac{ \omega_{g-\frac12,n+1}(p_0,J)}{dx(p_0)},\label{Rec}
        \end{align}
        and the $*$ in the sum denotes that we remove terms involving $\omega_{0,1}$.
\end{defin}

 As expected, it is shown in \cite{O23} that $\{\omega_{g,n+1}\}_{g,n}$ satisfies the Chekhov-Eynard-Orantin topological recursion when $\sQ=0$. However, it is important to remark that it is conjectural that the above definition makes sense for $2g-2+n\geq1$ when $\Sigma\neq\bP^1$ or $\sQ\neq0$ --- there is no issue when $2g-2+n=0$. In particular, it has not been proven whether the above formula constructs symmetric multidifferentials $\omega_{g,n+1}$ on $(\Sigma)^{n+1}$ --- the definition only ensures the well-definedness within a fundamental domain due to $\eta^p_\sA(p_0)$ in the formula. When $\Sigma=\bP^1$, \cite{KO22,O23} proved several properties on $\omega_{g,n+1}$ which are summarised as below:

\begin{thm}[\cite{KO22,O23}]\label{thm:RTR}
When $\Sigma=\bP^1$, $\omega_{g,n+1}$ are well-defined multidifferentials on $(\Sigma)^{n+1}$ and they satisfy the following properties:
\begin{itemize}
    \item $\omega_{g,n+1}$ are symmetric multidifferentials
    \item For $2g-2+n\geq0$, $\omega_{g,n+1}(p_0,J)$ has no residues as a differential in $p_0$, and their poles only lie in $\cR^*\cup\sigma(J\cup\cP_+^0)$.
    \item For $2g-2+n\geq0$, let $\phi$ be any primitive of $\omega_{0,1}$, then
    \begin{equation}
        (2-2g-n-1)\cdot \omega_{g,n+1}(J_0)=\frac{1}{2\pi i}\oint_{p\in C^{\mathfrak p}_-}\phi(p)\cdot\omega_{g,n+2}(p,J_0)\label{dilaton}
    \end{equation}
\end{itemize}
\end{thm}
\begin{conj}[\cite{O23}]
{\rm Theorem \ref{thm:RTR}} holds for any $\Sigma$.
\end{conj}

As discussed in \cite{O23}, it is easy to see for each $g,n$ that $\omega_{g,n+1}$ polynomially depends on $\sQ$ up to $\sQ^{2g}$, and the recursion for the $\sQ$-top degree part is self-closed, i.e. they can be constructed without the information of lower degree parts. We call it the $\sQ$-top recursion, and explicitly it is defined as follows:

\begin{defin}[\cite{O23}]
    Given a hyperelliptic refined spectral curve $\cS_{{\bm \mu}, {\bm \kappa}}$, the \emph{$\sQ$-top recursion} is a recursive definition of multidifferentials $\varpi_{g,n+1}$ on $(\Sigma)^{n+1}$ by the following formulae:
    \begin{align}
        \varpi_{0,1}(p_0):&=y(p_0)\cdot dx(p_0),\label{pi01}\\
    \varpi_{0,2}(p_0,p_1):&=-B(p_0,\sigma(p_1)),\label{pi02}\\
    \varpi_{\frac12,1}(p_0):&=\frac{1}{2}\left(-\frac{d\Delta y(p_0)}{\Delta y(p_0)}+\sum_{p\in\widetilde{\cP}_+}\mu_p\cdot\eta^p_\sA(p_0)+\sum_{i=1}^{\tilde g}\kappa_i\cdot \int_{\cB_i}B(\cdot,p_0)\right),\label{pi1/2,1}
    \end{align}
    and for $2g-2+n\geq0$,
      \begin{equation}
            \varpi_{g,n+1}(J_0):=\frac{1}{2 \pi i}\left(\oint_{p\in C_+}-\oint_{p\in C_-}\right)\frac{\eta_\sA^p(p_0)}{4\omega_{0,1}(p)}\cdot {\rm Rec}_{g,n+1}^{\sQ\text{-}{\rm top}}(p,J),\label{Qtop}
        \end{equation}
        where 
        \begin{align}
            {\rm Rec}_{g,n+1}^{\sQ\text{-}{\rm top}}(p_0;J):=&\sum^*_{\substack{g_1+g_2=g\\J_1\sqcup J_2=J}} \varpi_{g_1,n_1+1}(p_0,J_1)\cdot \varpi_{g_2,n_2+1}(p_0,J_2)\nonumber\\
    &+\sum_{t\sqcup I=J}\frac{dx(p_0)\cdot dx(t)}{(x(p_0)-x(t))^2}\cdot \varpi_{g,n}(p_0,I)+ dx \cdot d_0\frac{ \varpi_{g-\frac12,n+1}(p_0,J)}{dx(p_0)}.\label{QRec}
        \end{align}
\end{defin}

Note that there is no $\varpi_{g-1,n+2}$ in $Q_{g,n+1}^{\sQ\text{-}{\rm top}}$, unlike $Q_{g,n+1}^{\sQ}$. Since the $\sQ$-top recursion is a subsector of the refined topological recursion, Theorem \ref{thm:RTR} holds for $\varpi_{g,n+1}$ too, as long as $\Sigma=\bP^1$. We note that it is meaningful to define the $\sQ$-top recursion independently and study it on its own. For example, as discussed in \cite{O23}, the $\sQ$-top recursion would be relevant to the Nekrasov-Shatashivili limit which is an active research area in mathematics and physics. In particular, \cite{O23} proved the following property for any $\Sigma$, not limited to $\Sigma=\bP^1$:
\begin{thm}[\cite{O23}]\label{thm:Qtop}
    $\varpi_{g,1}$ are well-defined residue-free differentials on $\Sigma$ whose poles only lie in $\cR^*\cup\sigma(\cP_+^0)$, and there exists an ordinary second order differential equation of the following form:
\begin{equation}
    \left( \epsilon_1^2\frac{d^2}{dx(p)^2}-Q_0(x(p))-\sum_{k\in\bZ_{\geq1}}\epsilon_1^{k}Q_k(x(p))\right) \psi^{\sQ\text{-}{\rm top}}(p)=0\label{QC}
\end{equation}
    where $Q_k(x)$ is a rational function of $x$ explicitly constructed by $\varpi_{h,1}$ for $2h<k$, and $\psi^{\sQ\text{-}{\rm top}}$ is a formal series in $\epsilon_1$ defined by
    \begin{equation}
        \epsilon_1 \cdot d\log\psi^{\sQ\text{-}{\rm top}}(p):=\sum_{g\geq0}\epsilon_1^{2g}\cdot \varpi_{g,1}(p).
    \end{equation}
\end{thm}
The associated differential operator \eqref{QC} is called the \emph{$\sQ$-top quantum curve}. Except for a special class of genus-zero curves investigated in \cite{KO22}, existence of the refined quantum curve in full generality is still an open question.

When the underlying hyperelliptic curve depends on complex parameters $\bm t=\{t_1,..,t_n\}$, one can consider a $\bm t$-parameter family $\cS_{\bm\kappa,\bm\mu}(\bm t)$ of refined spectral curves  as long as $\bm t$ are in a domain such that no points in $\cR\cup\cP$ collide. All the above definitions and theorems hold for $\cS_{\bm\kappa,\bm\mu}(\bm t)$. In the next section, we will consider how $\omega_{g,n+1}(\bm t)$ behave while one varies $\bm t$.

Before turning to the variational formula, let us define the free energy $F_g$, except $F_0,F_{\frac12},F_1$ which will be defined later:
\begin{defin}[\cite{KO22,O23}]
    For $g>1$, the \emph{genus-$g$ free energy} $F_g,F_g^{\sQ\text{-}{\rm top}}$ of the refined topological recursion and the $\sQ$-top recursion is defined respectively as follows:
    \begin{align}
        F_g:&=\omega_{g,0}:=\frac{1}{2-2g}\frac{1}{2\pi i}\oint_{p\in C^{\mathfrak p}_-}\phi(p)\cdot\omega_{g,1}(p),\\
        F_g^{\sQ\text{-}{\rm top}}:&=\varpi_{g,0}:=\frac{1}{2-2g}\frac{1}{2\pi i}\oint_{p\in C^{\mathfrak p}_-}\phi(p)\cdot\varpi_{g,1}(p).
    \end{align}
\end{defin}

\newpage
\section{Variation}\label{sec:Var}
The variational formula is proven in \cite{EO07} and originally it is explained as follows. Consider a one-parameter family of spectral curves $\cS(t)$ in the unrefined setting. Then, $x$ and $y$ as functions on $\Sigma$ depend on the parameter $t$ and so do all $\omega_{g,n+1}(t)$. Then, \cite{EO07} considers a special type of deformation, namely, \emph{variation for fixed $x$}. This may sound contradictory with the fact that $x$ depends on $t$, but what it really means is the following.

Set $\sQ=0$. By choosing one of the branched sheet, one projects $\omega_{g,n+1}$ down to $\bP^1$ away from ramification points and treat them locally as multidifferentials on $\bP^1$. The variation for fixed $x$ means that we apply the \emph{partial} derivative with respect to $t$ for these multidifferentials on $\bP^1$ with the understanding that $\frac{\partial }{\partial t}dx_a=0$, and apply the local inverse $x^{-1}$ to pull them back to differentials on $\Sigma$. That is, the variation symbol $\delta^{\text{EO}}_{t}$ in \cite{EO07} acting on $\omega_{g,n+1}$ means (c.f. \cite{IKT18-1,BCCGF22}):
\begin{equation}
\delta^{\text{EO}}_{t}*\omega_{g,n+1}(p_0,..,p_n;t):=\left(\frac{\partial}{\partial  t}\omega_{g,n+1}(p_{t}(x_0),..,p_{t}(x_n);t)\right)\bigg|_{x_a=x(p_a)},
\end{equation}
where on the right-hand side we think of $x$ as independent of $t$ and instead $p_t$ depends on both $t$ and $x$. We will denote by $*$ the action of the variation in order to distinguish from the standard product symbol $\cdot$ which we are using throughout the paper. The standard partial derivative notation $\partial_t$ is commonly used in e.g. \cite{E17,EGF19,MO19} but we avoid this notation to emphasise that the operation is not just a partial derivative.

We will provide another equivalent description of the variation operation without considering the projection and inverse. The motivation of introducing such a new perspective is for the clarity of the proof of the variational formula when $\sQ\neq0$. The original proof by Eynard and Orantin is based on a graphical interpretation whose analogue does not exist in the refined setting, at least at the moment of writing. As a consequence, we need to directly evaluate the variation of the refined recursion formula \eqref{RTR}, and in this case, taking the projection and the inverse becomes subtle because $C_\pm$ contains $J_0$ and $\sigma(J_0)$. 

\begin{defin}\label{def:var}
Given $\cS_{\bm \mu,\bm\kappa}(t)$, the \emph{topological recursion variational operator} $\delta_{t}^{(n)}$ is a differential operator acting on meromorphic functions on $(\Sigma)^n$ defined by
\begin{equation}
    \delta_{t}^{(n)}:=\frac{d}{d t}-\sum_{a=1}^n\frac{\partial x(p_a)}{\partial  t}\frac{1}{dx(p_a)}d_{p_a},\label{d}
\end{equation}
where $(p_1,..,p_n)\in(\Sigma\backslash\cR)^n$ and $d_{p_a}$ denotes the exterior derivative with respect to $p_a$. We extend the action of $\delta_{t}^{(n)}$ to a meromorphic multidifferential $\omega$ on $(\Sigma)^n$ by
\begin{equation}
\delta_{t}^{(n)}*\omega(p_1,..,p_n;t):=\left(\delta_{t}^{(n)}*\frac{\omega(p_1,...,p_n;t)}{dx(p_1)\cdots dx(p_n)}\right)\cdot dx(p_1)\cdots dx(p_n).
\end{equation}
\end{defin}

Note that this definition is valid not only for hyperelliptic curves but also for any algebraic curves. It can be generalised to a multi-parameter family in an obvious way. $\delta_{t}^{(n)}$ is defined only when each $p_a\not\in\cR$ which resonates with the fact that one has to choose a branch in the Eynard-Orantin description. Importantly, the above definition implies
\begin{equation}
    \delta_{t}^{(1)}* x=0,\quad\delta_{t}^{(1)}* dx=0,
\end{equation}
and for a differential $w$ on $(\bP^1)^n$, its pullback to $(\Sigma)^n$ satisfies
\begin{equation}
    \delta_{t}^{(n)}* w(x(z_1),...,x(z_n);t)=\frac{\partial}{\partial t}w(x(z_1),...,x(z_n);t).
\end{equation}
Thus, $\delta_{t}^{(n)}$ in fact serves as the variation for fixed $x$. Furthermore, we have
\begin{equation}
    \delta_{t}^{(1)}* ydx=\frac{\partial y}{\partial t}dx-\frac{\partial x}{\partial t}dy,\label{var01}
\end{equation}
which corresponds to \cite[Equation 5-2]{EO07}. From now on, we omit writing the $t$-dependence of functions and multidifferentials.

\begin{rem}
    Perhaps, the conceptual motivation of the action of $\delta^{(n)}_t$ becomes clearer when one thinks of the underlying hyperelliptic curve from the Hitchin perspective \cite{DM13,DM15,E17}. A Hitchin spectral curve (of rank 2) is given by a triple $(\Sigma^o,\varphi,\pi)$ where $\pi:\Sigma^o\to\bP^1$, $\varphi$ is a quadratic differential on $\bP^1$, and $\Sigma^o$ is embedded in $T^*\bP^1$ as
    \begin{equation}
        \Sigma^o=\{\lambda\in T^*\Sigma^o|\lambda^{\otimes2}=\pi^*\varphi\}\subset T^*\bP^1.
    \end{equation}
    Our $\Sigma$ would be obtained after normalisation and compactification of $\Sigma^o$. By interpreting $\pi=x$ and $\varphi=(ydx)^{\otimes2}$, \emph{variation for fixed} $x$ means that one varies the quadratic differential $\varphi$ while keeping the projection $\pi=x$ invariant.
\end{rem}

Given an unrefined spectral curve $\cS(t)$, let us assume existence of a pair $(\gamma,\Lambda)$ such that $\gamma$ is a path in $\Sigma\backslash\cR$ and $\Lambda$ is a function holomorphic along $\gamma$ satisfying
    \begin{equation}
        \delta_{t}^{(1)}* \omega_{0,1}(p_1)=:\int_{p\in\gamma}\Lambda(p)\cdot\omega_{0,2}(p,p_1).\label{varcondition1}
    \end{equation}
Then, \cite{EO07} showed that the following relation holds for $g,n\in\bZ_{\geq0}$ by using the graphical interpretation of the unrefined topological recursion formula, which is known as the \emph{variational formula}:
\begin{equation}
    \delta_t^{(n+1)}*\omega_{g,n+1}(J_0)=\int_{p\in\gamma}\Lambda(p)\cdot \omega_{g,n+2}(p,J_0).\label{EOvar}
\end{equation}

The difficulty to generalise the variational formula into the refined setting arises due to the more complicated pole structure of $\{\omega_{g,n+1}\}_{g,n}$. Nevertheless, if we restrict the pair $(\gamma,\Lambda)$ to certain classes as below, a refined analogue still holds when $\Sigma=\bP^1$, and we expect that it works for any $\Sigma$ in general.

For $s\in\cP^{\infty}\backslash\cR$ and $r\in\cP^{\infty}\cap\cR$, let $x(s)=x_s,x(r)=x_r$ and suppose $\omega_{0,1}$ behaves locally
\begin{equation}
    \omega_{0,1}=\pm\left(\sum_{k=0}^{m_s}\frac{t_{s,k}}{(x-x_s)^{k+1}}+\cO(1)\right)dx,\quad \omega_{0,1}=\left(\sum_{k=1}^{m_r}\frac{t_{r,k}}{(x-x_r)^k}+\cO(1)\right)\frac{dx}{2\sqrt{x-x_r}}\label{var-para}
\end{equation}
Let $\Lambda_{s,k},\Lambda_{r,k}$ be the corresponding meromorphic function on $\Sigma$ such that
\begin{equation}
    \frac12\left(\Res_{p=s}-\Res_{p=\sigma(s)}\right)\Lambda_{s,k}(p)^{-1}\cdot\omega_{0,1}(p)=t_{s,k},\quad\Res_{p=r}\Lambda_{r,k}(p)^{-1}\cdot\omega_{0,1}(p)=t_{r,k}.\label{Lambda}
\end{equation}
\cite{EO07,E17} show a construction of each $\Lambda_{s,k},\Lambda_{r,k}$, at least locally. Note that their pole is at most of order $m_s-1,m_r-1$ respectively.

\begin{defin}[\cite{E17}]\label{def:pair}
Given $\cS_{\bm\kappa,\bm\mu}(\bm t)$, $(\gamma,\Lambda)$ is said to be a \emph{generalised cycle} if it falls into one of the following kinds:
\begin{itemize}
    \item[\textbf{I}\,\,\,\,\,:] $\gamma\in\{\cB_i\}_{i\in\{1,..,\tilde g\}}$ and $\Lambda=1$
    \item[\textbf{II}\,\,\;:] Let $p\in\Sigma$ be an $m_p$-th order pole of $\omega_{0,1}$ where $m_p\geq2$. Then, for $k\in\{1,..,m_p-1\}$, $\Lambda_{p,k}$ is given as in \eqref{Lambda}, and $\gamma_{p,k}$ is a union of contours encircling $p$ and $\sigma(p)$ in the opposite orientation if $p\not\in\cR$, and $\gamma_{p,k}$ is a contour encircling $p$ if $p\in\cR$.
    \item[\textbf{III}\;:] Let $p\in\Sigma$ be a location of a residue of $\omega_{0,1}$ which necessarily means $p\not\in\cR$. Then, $\gamma_p$ is an open path from $\sigma(p)$ to $p$ within a fundamental domain, and $\Lambda_p=1$.
\end{itemize}
The corresponding parameters $t_{(\gamma,\Lambda)}$ defined by the expansion \eqref{var-para} are called \emph{2nd kind times} or \emph{3rd kind times}, whereas \emph{1st kind times} are defined by
\begin{equation}
    t_i:=\frac{1}{2\pi i}\oint_{\cA_i}\omega_{0,1},\label{t1}
\end{equation}
\end{defin}

1st, 2nd, and 3rd kind times are respectively called filling fractions, temperatures, and moduli of the poles in \cite{EO07}. All generalised cycles $(\gamma,\Lambda)$ are anti-invariant under $\sigma$ when it applies to integration. 2nd and 3rd kind times are often refered to as \emph{KP times} and their relation to KP systems are discussed in \cite{E17}.

We consider a refined spectral curve $\cS_{\bm\kappa,\bm\mu}(\bm t)$ such that $t_1,..,t_{|\bm t|}\in\bm t$ are defined as above, which are independent of each other, and we denote by $(\gamma_1,\Lambda_1),..,(\gamma_{|\bm t|},\Lambda_{|\bm t|})$ associated generalised cycles. In this setting, the variational formula \eqref{EOvar} holds in the unrefined setting as shown in \cite{EO07}. However, when $\sQ\neq0$, it turns out that an analogous statement holds if $\cS_{\bm\kappa,\bm\mu}(\bm t)$ satisfies an additional condition, which we call the refined deformation condition:
\begin{defin}\label{def:condi}
    Consider $\cS_{\bm\kappa,\bm\mu}(\bm t)$ parameterised by times of the 1st, 2nd, and 3rd kind $\bm t=(t_1,..,t_{|\bm t|})$. We say that $\cS_{\bm\kappa,\bm\mu}(\bm t)$ satisfies the \emph{refined deformation condition} with respect to $t_l$ for $l\in\{1,..,{|\bm t|}\}$ if the following holds:
    \begin{equation}
         \delta_{t_l}^{(1)}* \omega_{\frac12,1}(p_1)=\int_{q\in\gamma_l}\Lambda_l(q)\cdot\omega_{\frac12,2}(q,p_1).\label{rdc}
    \end{equation}
    We say that $\cS_{\bm\kappa,\bm\mu}(\bm t)$ satisfies the \emph{refined deformation condition} if the above holds for all $l$.
\end{defin}
Note that in the unrefined setting the variational formula \eqref{EOvar} for $(g,n)=(0,1)$ automatically holds if $\omega_{0,2}=B$. Even if $\omega_{0,2}$ is defined differently, it is then observed in e.g. \cite{BCCGF22} that the variational formula still works for the rest of $\omega_{g,n+1}$, as long as the variational relation \eqref{EOvar} holds for $(g,n)=(0,1)$. In other words, it has to be rather \emph{imposed} as a supplemental condition in addition to \eqref{varcondition1}. The refined deformation condition (Definition \ref{def:condi}) is analogous to this observation.

Finally, we will state the variational formula in the refined setting, whose proof is entirely given in Appendix \ref{app:varproofw} and \ref{app:varproofF} because it is lengthy:
\begin{thm}\label{thm:var}
    When $\Sigma=\bP^1$, assume that $\cS_{\bm\kappa,\bm\mu}(\bm t)$ satisfies the refined deformation condition with respect to $t_l$ for $l\in\{1,..,|\bm t|\}$. Then, $\omega_{g,n+1}$ and $F_g$ ($g>1$ for $F_g$) satisfy:
    \begin{equation}
        \frac{\partial F_g}{\partial t_l}= \int_{p\in\gamma_l}\Lambda_l(p)\cdot \omega_{g,1},\qquad\delta_{t_l}^{(n+1)}*\omega_{g,n+1}(J_0)=\int_{p\in\gamma_l}\Lambda_l(p)\cdot \omega_{g,n+2}(p,J_0).\label{var}
    \end{equation}
\end{thm}
\begin{conj}
    {\rm Theorem \ref{thm:var}} holds for any $\Sigma$.
\end{conj}

\newpage
\section{Examples}\label{Sec:App}

We will now apply the variational formula to several examples 

\subsection{Hypergeometric type curves}
Hypergeometric type curves are the classical limit of a confluent family of Gauss hypergeometric differential equations, and they are discussed in \cite{IKT18-1,IKT18-2,IK20,IK21} in relation to the BPS invariants and Stokes graphs. Hypergeometric type curves are classified into nine types based on their pole structure, and seven of them depend on parameters. \cite{IKT18-2} already write all the seven types of curves in terms of 3rd kind times, which they denote by $m_p$ rather than $t_p$. Then, the question one should ask is whether the corresponding refined spectral curve $\cS_{\bm\mu}(\bm t)$ satisfies the refined deformation condition. Hypergeometric type curves are main examples considered in \cite{KO22}.
\begin{prop}\label{prop:rdc1}
    Every refined spectral curve $\cS_{\bm\mu}(\bm t)$ associated with a hypergeometric type curve in the form of \cite{IKT18-2} satisfies the refined deformation condition.
\end{prop}
\begin{proof}
    The proof is done by explicit computations. Since they are genus-zero curves, a rational expression of $x,y$ is given in e.g. \cite{IKT18-2} in terms of a coordinate $z$ on $\bP^1$, from which one can construct the variational operator $\delta_t^{(1)}$ for all $t\in\bm t$ with respect to $z$. Then, all one has to do is to compute $\omega_{\frac12,1}(z_0)$ and $\omega_{\frac12,2}(z_0,z_1)$ from the refined topological recursion and explicitly check the refined deformation condition. See Appendix \ref{app:rdc} where we present explicit computations for a few examples.
\end{proof}

One can use the variational formula as the defining equation for $F_\frac12$ and $F_1$ as follows --- since all $\omega_{0,n}$ is independent of the refinement parameter $\sQ$, we can define $F_0$ as \cite{EO07} does:  
\begin{defin}\label{def:FHG}
    For a refined spectral curve $\cS_{\bm\mu}(\bm t)$ associated with a hypergeometric type curve, $F_\frac12$ and $F_1$ are defined as a solution of the following differential equations for all $k,l\in\{1,..,|\bm t|\}$:
    \begin{align}
        \frac{\partial F_\frac12}{\partial t_k\partial t_l}:=\int_{p_1\in\gamma_k}\int_{p_2\in\gamma_l}\omega_{\frac12,2}(p_1,p_2),\qquad
        \frac{\partial F_1}{\partial t_k}:=\int_{p_1\in\gamma_k}\omega_{1,1}(p_1),\label{F1}
    \end{align}
    where $F_\frac12$ is defined up to linear terms in $t_l$ and $F_1$ is defined up to constant terms.
\end{defin}
Since $\Lambda=1$ for the 3rd kind, we immediately obtain the following:
\begin{cor}\label{cor:var}
    For a refined spectral curve $\cS_{\bm\mu}(\bm t)$ associated with a hypergeometric type curve, we have the following for $2g-2+n\geq1$:
    \begin{equation}
        \prod_{a=1}^n\frac{\partial}{\partial t_{l_a}} F_g=\prod_{a=1}^n\int_{p_a\in\gamma_{l_a}}\omega_{g,n}(p_1,..,p_n).\label{dFdt}
    \end{equation}
\end{cor}

Corollary \ref{cor:var} becomes useful to derive a relation between refined BPS structures \cite{Bri19-1,BBS19} and the refined topological recursion, as a generalisation of \cite{IKT18-1,IKT18-2,IK20}. For a general refined spectral curve $\cS_{\bm\kappa,\bm\mu}(\bm t)$, not limited to hypergeometric type curves, we will define $F_\frac12,F_1$ in a similar way to Definition \ref{def:FHG}. See Remark \ref{rem:F1/2}.

\subsection{A degenerate elliptic curve}
Let us consider the case where $x$ and $y$ satisfy the following algebraic equation:
\begin{equation}
     y^2-Q_0(x)=0,\quad Q_0(x):=4\left(x-q_0\right)^2\left(x+2q_0\right)=4x^3+2tx+8q_0^3,\quad q_0=\sqrt{-\frac{t}{6}}.\label{curve0}
\end{equation}
A convenient rational expression of $x,y$ in terms of a coordinate $z$ on $\Sigma=\bP^1$ is
\begin{equation}
    x(z)=z^2-2q_0,\quad y(z)=2z(z^2-3q_0)=2z(z^2-q_z^2),\label{rational}
\end{equation}
where for brevity, we set $q_z:=\sqrt{3q_0}$. It appears in a singular limit (as an algebraic curve) of the following elliptic curve,
\begin{equation}
    y^2=4x^3-g_2x-g_3,\label{curve1}
\end{equation}
where for generic $g_2,g_3$ we can write $x,y$ in terms of the Weierstrass $\wp$-function as $x=\wp$ and $y=\wp'$. In \cite{IS15,Iwa19}, the curve \eqref{curve0} or \eqref{curve1} is chosen as a spectral curve of the Chekhov-Eynard-Orantin topological recursion, and a relation between the free energy and a $\tau$-function of the Painlev\'{e} I equation is proven.

With the above parameterisation, the hyperelliptic involution $\sigma$ acts as $\sigma:z\mapsto-z$, and $\cR=\{0,\infty\}$ with $\cR^*=\{0\}$. Note that $\omega_{0,1}(z)$ has a simple zero at $z=\pm q_z$, hence we choose $\cP_+=\{q_z\}$ and we assign $\mu\in\bC$ to $z=q_z$. Since $H_1(\Sigma,\bZ)=0$ in this example, the above choice uniquely defines a refined spectral curve $\cS_{\mu}(t)$. Theorem \ref{thm:RTR} then implies that $\omega_{g,n+1}(z_0,J)$ have poles, as a differential in $z_0$, at $z_0=0,-z_1,..,-z_n,-q_z$ when $2g-2+n\geq0$. 

As shown in \cite{IS15}, $t$ in \eqref{curve0} plays the role of a 2nd kind time, and the corresponding generalised cycle can be decoded from the following equations
\begin{equation}
   \Lambda_t(z):=-z+\frac{c q_0}{z},\quad \Res_{z=\infty}\Lambda_t(z)^{-1}\cdot\omega_{0,1}(z)=t,\quad\delta^{(1)}_t*\omega_{0,1}(z_0)=\Res_{z=\infty}\cdot \Lambda_t(z)\cdot\omega_{0,2}(z,z_0),\label{Painleve-t}
\end{equation}
where $c$ is one of the roots of $2c^2-6c+3=0$. The second term in $\Lambda_t$ is irrelevant in the last equation in \eqref{Painleve-t}, and it is indeed absent in \cite{IS15}, though it is necessary for the second equation. Now one may ask: does every $\cS_{\mu}(t)$ satisfy the refined deformation condition similar to hypergeometric type curves (Proposition \ref{prop:rdc1})? Here is the answer to that question:
\begin{prop}\label{prop:rdc2}
    Let $\cS_{\mu}(t)$ be a refined spectral curve defined as above. Then, it satisfies the refined deformation condition if and only if $\mu=1$.
\end{prop}
\begin{proof}
    The proof is again by explicit computations, similar to Proposition \ref{prop:rdc1}. That is, we explicitly write the variational operator $\delta_t^{(1)}$ in terms of $t$ and $z$, and confirm when \eqref{rdc} is satisfied. Since everything can be expressed as rational functions, it is easy to find that $\mu=1$ is the only solution. See Appendix \ref{app:rdc} for computations.
\end{proof}
Note that, unlike $\omega_{g,n+1}$ for $2g-2+n\geq0$, poles of $\omega_{\frac12,1}(z_0)$ are all simple and they are located not only at $z_0=0,-q_z$ but also at $z_0=q_z,\infty$ whose residues are given as:
\begin{equation}
    \Res_{z=0}\omega_{\frac12,1}(z)=-\frac{\sQ}{2},\quad \Res_{z=\infty}\omega_{\frac12,1}(z)=\frac{3\sQ}{2},\quad \Res_{z=\pm q_z}\omega_{\frac12,1}(z)=\frac{\sQ}{2}(-1\pm\mu)\label{mu}
\end{equation}
Therefore, the refined deformation condition is satisfied exactly when $\omega_{\frac12,1}$ becomes regular at $\cP_+$. Even if we choose $\cP_+=\{-q_z\}$ instead, this aspect remains correct. That is, the refined deformation condition for this curve is equivalent to the condition such that $\omega_{\frac12,1}$ becomes regular at $\cP_+$, no matter how $\cP_+$ is chosen.

\subsubsection{$\sQ$-top quantum curve}
Theorem \ref{thm:Qtop} shows that the $\sQ$-top recursion can be utilised to quantise a refined spectral curve. For a general refined spectral curve $\cS_{\bm\kappa,\bm\mu}(\bm t)$, not limited to the above example, we introduce the following terminology:
\begin{defin}\label{def:QQC}
    We say that a refined spectral curve $\cS_{\bm\kappa,\bm\mu}(\bm t)$ satisfies the \emph{$\sQ$-top quantisation condition} if for each $k$ the set of poles of $Q_{k\geq1}^{\sQ\text{-{\rm top}}}$ is a subset of that of $Q_0^{\sQ\text{-{\rm top}}}$.
\end{defin}
We return to our example, and consider the $\sQ$-top quantisation condition for $\cS_{\mu}(t)$.
\begin{prop}\label{prop:QQC}
    The above refined spectral curve $\cS_\mu(t)$ satisfies the $\sQ$-top quantisation condition if and only if $\mu=1$. 
\end{prop}

\begin{proof}
    The proof is again by computations. The formula in \cite{O23} gives
    \begin{equation}
    Q_1^{\sQ\text{-{\rm top}}}(z_0):=\frac{\varpi_{0,1}(z_0)}{dx(z_0)^2}\cdot\mu \cdot \eta^{q_z}(z_0)=2q_z\cdot\mu,\label{Q1}
\end{equation}
\begin{equation}
    Q_{k\geq2}^{\sQ\text{-{\rm top}}}(z_0):=\frac{2\varpi_{0,1}(z_0)\cdot R^{\sQ\text{-{\rm top}}}_{\frac{k}{2},1}(p_0)}{dx(p_0)\cdot dx(p_0)},\qquad R^{\sQ\text{-{\rm top}}}_{\frac{k}{2},1}(z_0)=\Res_{z=q_z}\frac{\eta^z(z_0)}{2\omega_{0,1}(z)}\cdot{\rm Rec}_{\frac{k}{2},1}^{\sQ\text{-{\rm top}}}(z).\label{Qk>1}
\end{equation}
The if part is easy to see. By setting set $\mu=1$, then \eqref{mu} implies that $\omega_{\frac12,1}$ becomes regular at $z=q_z$ hence $Q_{k\geq2}^{\sQ\text{-{\rm top}}}$ becomes regular at $x=q_0$. See Appendix \ref{app:rdc} for the only-if part. 
\end{proof}
Therefore, \emph{the refined deformation condition and the $\sQ$-top quantisation condition agree} for this example. Note that any refined spectral curve of hypergeometric type satisfies the $\sQ$-top, and in fact the refined quantisation condition. We expect that no additional condition will appear in the full refined quantisation, and it is interesting to see whether this coincidence holds for other curves, e.g. curves related to other Painlev\'{e} equations \cite{IMS16}.

To close, we prove that the $\sQ$-top quantum curve for $\cS_{\mu=1}(t)$ is written in terms of the $\sQ$-top free energy $F_g^{\sQ\text{-{\rm top}}}$ whose proof will be given in Appendix \ref{sec:proof}. \cite{LN21,LN22} discuss a similar equation in the context of accessory parameters and conformal blocks in the Nekrasov-Shatashivili limit. Thus, we conjecture that the $\sQ$-top free energy $F_g^{\sQ\text{-{\rm top}}}$ coincides with the Nekrasov-Shatashivili effective twisted superpotential \cite{NS09} even when $\Sigma\neq\bP^1$ as long as an appropriate refined spectral curve is chosen.

\begin{thm}\label{thm:QC}
    For $\cS_{\mu=1}(t)$ described above, the $\sQ$-top quantum curve is given as:
   \begin{equation}
    \left( \epsilon_1^2\frac{d^2}{dx(p)^2}-4x^3-2tx-2\sum_{g\in\frac12\bZ_{\geq0}}\epsilon_1^{2g}\frac{\partial F_g^{\sQ\text{-{\rm top}}}}{\partial t}\right) \psi^{\sQ\text{-}{\rm top}}(p)=0,\label{FgQC}
\end{equation}
where $F_{\frac12}^{\sQ\text{-{\rm top}}}$ and $F_{1}^{\sQ\text{-{\rm top}}}$ are defined as a solution of the following differential equation:
\begin{equation}
   \frac{\partial^2}{\partial t^2}F_{\frac12}^{\sQ\text{-{\rm top}}}=\Res_{z_1=0}\Res_{z_0=0}\cdot \Lambda_t(z_1)\cdot\Lambda_t(z_0)\cdot\omega_{\frac12,2}(z_0,z_1),\quad F_\frac12^{\sQ\text{-{\rm top}}}\big|_{t=0}=\frac{\partial}{\partial t}F_\frac12^{\sQ\text{-{\rm top}}}\bigg|_{t=0}=0.
\end{equation}
\begin{equation}
   \frac{\partial}{\partial t}F_{1}^{\sQ\text{-{\rm top}}}=\Res_{z_0=0}\cdot \Lambda_t(z_0)\cdot\omega_{1,1}(z_0),\quad F_1^{\sQ\text{-{\rm top}}}\big|_{t=1}=0.
\end{equation}
\end{thm}

\newpage
\appendix
\section{Proofs}\label{sec:proof}
Throughout Appendix, we set $\Sigma=\bP^1$. We will give detailed computations for most of propositions and theorems of the present paper.

\subsection{Proof of Theorem \ref{thm:var}: for $\omega_{g,n+1}$}\label{app:varproofw}
We assume that a refined spectral curve $\cS_{\bm\mu}(t)$ carries one time $t$ of either 2nd or 3rd kind, and we denote by $(\gamma,\Lambda)$ the associated generalised cycle. The arguments below can be easily generalised to curves with several times.

Let us first introduce convenient notations. First, for any multidifferential $\omega(p,J)$, we denote its anti-invariant part under $\sigma$ by
\begin{equation}
    \Delta_p\omega(p,J):=\omega(p,J)-\omega(\sigma(p),J),
\end{equation}
where the subscript shows the variable we are considering for the above operation. Next, in order to specify variables for the variational operator, we sometime use the following notation
\begin{equation}
    \delta_t^{(p_1,..,p_n)}=\delta_t^{(n)}=\frac{d}{d t}-\sum_{a=1}^n\frac{\partial x(p_a)}{\partial  t}\frac{1}{dx(p_a)}d_{p_a}.
\end{equation}
Then, we can extend the action of the variational operator to meromorphic functions on $(\Sigma)^m$ for $m\neq n$ without any issue.

\subsubsection{Useful lemmas}
We show how the variational operator $\delta_t^{(n)}$ behaves on a product of functions and differentials:
\begin{lem}\label{lem:chain}
    Let $f(p,p_0)$ be a meromorphic function of $p$ and differential in $p_0$ and $\omega(p,p_1)$ a meromorphic bidifferential on $\Sigma$. Then, for any $o\in\Sigma$, we have the following:
    \begin{equation}
        \delta_t^{(p,p_0,p_1)}*\big(f(p,p_0)\cdot\omega(p,p_1)\big)= \delta_t^{(p,p_0)}*\big(f(p,p_0)\big)\cdot\omega(p,p_1)+ f(p,p_0)\cdot\delta_t^{(p,p_1)}*\omega(p,p_1)\label{Leibniz}
    \end{equation}
    \begin{equation}
    \delta_t^{(p_0,p_1)}*\Res_{p=a}f(p,p_0)\cdot\omega(p,p_1)=\Res_{p=o}\delta_t^{(p,p_0,p_1)}*(f(p,p_0)\cdot\omega(p,p_1))\label{dRes}
\end{equation}
\end{lem}
\begin{proof}
    \eqref{Leibniz} is just a Leibniz rule for the variational operator, and it is straightforward.
    
    On the other hand, we need a more careful consideration to prove \eqref{dRes}. Let us first show that $\delta_t^{(p_0,p_1)}$ commutes with $\Res_{p=o}$, no matter if $o$ depends on $p_0$, $p_1$ or $t$. Let $z$ be local coordinates around $a$, and suppose the integrand of the left-hand side is expanded at $z(p)=z(o)$ as
    \begin{equation}
       f(p,p_0)\cdot\omega(p,p_1)= \sum_{k\in\bZ}h_k(p_0,p_1,o)\cdot\frac{dz(p)}{(z(p)-z(o))^{k+1}},
    \end{equation}
    where $h_k(p_0,p_1,o)$ are bidifferential in $p_0,p_1$. Then, after taking the residue, the left-hand side of \eqref{dRes} is simply
    \begin{equation}
       \text{L.H.S. of \eqref{dRes}}= \delta_t^{(p_0,p_1)}*h_0(p_0,p_1,o).
    \end{equation}
    On the other hand, since $z(p)$ can be thought of as a constant in terms of $\delta_t^{(p_0,p_1)}$, we find
    \begin{align}
        \delta_t^{(p_0,p_1)}*(f(p,p_0)\cdot\omega(p,p_1))=\sum_{k\in\bZ}\bigg(&\delta_t^{(p_0,p_1)}*(h_k(p_0,p_1,o))\cdot\frac{dz(p)}{(z(p)-z(o))^{k+1}}\nonumber\\
        &-h_k(p_0,p_1,o)\cdot\delta_t^{(p_0,p_1)}*(z(o))\cdot\frac{(k+1)dz(p)}{(z(p)-z(o))^{k+2}}\bigg),\label{dRes1}
    \end{align}
    where $\delta_t^{(p_0,p_1)}*(z(o))$ can be nonzero if $o$ depends on $p_0$, $p_1$, or $t$. Nevertheless, the second term in \eqref{dRes1} will have no contributions after taking the residue, and we have shown that $\delta_t^{(p_0,p_1)}$ commutes with $\Res_{p=o}$. One may interpret this result such that a closed contour encircling $p=o$ can be chosen independently from the time $t$.

    Our last task is to transform $\delta_t^{(p_0,p_1)}$ into $\delta_t^{(p,p_0,p_1)}$, that is, the variational operator becomes effective with respect to the variable of integration $p$ as well. In fact, by the chain rules, we find
    \begin{equation}
         \delta_t^{(p_0,p_1)}*(f(p,p_0)\cdot\omega(p,p_1))=\delta_t^{(p,p_0,p_1)}*(f(p,p_0)\cdot\omega(p,p_1))+d_p\left(f(p,p_1)\cdot\frac{\omega(p,p_1)}{dx(p)}\frac{\partial x(p)}{\partial t}\right).
    \end{equation}
    Then since $f$ and $\omega$ are both meromorphic, the last term vanishes after taking residue.
\end{proof}
Lemma \ref{lem:chain} can be easily generalised to $\delta_t^{(p_0,p_1,..,p_n)}$ for any $n$. We next recall useful results given in \cite{EO07} (see also \cite{Rau59}):
\begin{lem}[{\cite[Section 5.1]{EO07}}]\label{lem:varEO}
    For $\cS_{\bm\mu}(t)$, we have
    \begin{align}
        \delta_{\epsilon}^{(2)} * \omega_{0,2}(p_0,p_1)=&\delta_{\epsilon}^{(2)} * B(p_0,p_1)\nonumber\\
        =&-\sum_{r\in\cR}\Res_{p=r}\frac{\eta^p(p_0)}{4\,\omega_{0,1}(p)}\cdot\Bigl(B(p,p_1)-B(\sigma(p),p_1)\Bigr)\cdot\delta_{\epsilon}^{(1)}*\omega_{0,1}(p)\nonumber\\
        =&\int_{\gamma}\Lambda(p)\cdot\omega_{0,3}(p,p_0,p_1),\label{varB}\\
        \delta_{\epsilon}^{(2)} * \eta^p(p_0)=&\sum_{r\in\cR}\Res_{q=r}\frac{\eta^q(p_0)}{2\omega_{0,1}(q)}\cdot\eta^p(q)\cdot\delta_{\epsilon}^{(1)}*\omega_{0,1}(q)\nonumber\\
        =&-\frac{1}{2\pi i}\oint_{q\in C_+}\frac{\eta^q(p_0)}{\omega_{0,1}(q)}\cdot\eta^p(q)\cdot\delta_{\epsilon}^{(1)}*\omega_{0,1}(q),\label{vareta1},
    \end{align}
     where $p\in\Sigma$ is independent of $t$ and $p_0,..,p_n$ and $C_+$ is defined in Section \ref{sec:Def}.
\end{lem}
Note that, strictly speaking, \cite{EO07} only shows the first line of \eqref{vareta1}, and the second equality is a consequence due to \cite[Lemma 2.3]{O23} and the invariance of the integrand under $\sigma$ on $q$. With this property, we will show another lemma which is equivalent to e.g. \cite[Lemma 3.14]{BCCGF22}:
\begin{lem}\label{lem:dRec}
    Let $\omega(p;p_1,..,p_n)$ a meromorphic quadratic differential in $p$ and multidifferential in $p_1,..,p_n$ for some $n\in\bZ_{\geq0}$. Then, we have
    \begin{align}
        &\frac{1}{2\pi i}\int_{p\in C_+}\omega(p;J)\cdot\delta_{t}^{(2)}*\left(\frac{\eta^p(p_0)}{2\omega_{0,1}(p)}\right)\nonumber\\
        &=\frac{1}{2\pi i}\int_{p\in C_+}\frac{\eta^p(p_0)}{2\omega_{0,1}(p)}\cdot2\delta_{\epsilon}^{(1)}*\omega_{0,1}(p)\cdot\left(\frac{1}{2\pi i}\int_{q\in C_+}\frac{\eta^q(p)}{2\omega_{0,1}(q)}\cdot\omega(q;J)\right).\label{dRes2}
    \end{align}
\end{lem}
\begin{proof}
    Let us focus on the contribution from the action of $\delta_{t}^{(2)}$ on $\eta^p(p_0)$. Thanks to Lemma \ref{lem:varEO}, the corresponding term becomes
    \begin{equation}
       \left(\frac{1}{2\pi i}\right)^2\int_{p\in C_+}\int_{q\in C_+}\frac{\omega(p;J)}{2\omega_{0,1}(p)}\cdot\frac{\eta^q(p_0)}{\omega_{0,1}(q)}\cdot\eta^p(q)\cdot\delta_{\epsilon}^{(1)}*\omega_{0,1}(q),\label{dRes3}
    \end{equation}
    where $C_+$ with respect to $q$ contains $q=p$ inside. We now exchange the order of residues as follows (c.f. \cite[Appendix A]{EO07}, \cite[Appendix A.1]{O23})
\begin{equation}
    \int_{p\in C_+}\int_{q\in C_+}=\int_{q\in C_+}\left(\int_{p\in C_+}-2\pi i\Res_{p=q}\right),
\end{equation}
where $C_+$ with respect to $p$ on the right-hand side contains $p=q$ inside. Thus, we have
\begin{align}
    \eqref{dRes3}=&\frac{1}{2\pi i}\int_{q\in C_+}\frac{\eta^q(p_0)}{2\omega_{0,1}(q)}\cdot2\delta_{\epsilon}^{(1)}*\omega_{0,1}(q)\cdot\left(\frac{1}{2\pi i}\int_{p\in C_+}\frac{\eta^p(q)}{2\omega_{0,1}(p)}\cdot\omega(p;J)\right)\nonumber\\
    &+\frac{1}{2\pi i}\int_{q\in C_+}\omega(q;J)\cdot\frac{\eta^q(p_0)}{2\omega_{0,1}(q)^2}\cdot\delta_{t}^{(2)}*\omega_{0,1}(q)\label{dRes4}
\end{align}
After relabeling $p\leftrightarrow q$, one notices that the second term in \eqref{dRes4} precisely cancels the contribution of the action of $\delta_{t}^{(1)}$ on $\omega_{0,1}(p)$ on the left-hand side of \eqref{dRes2}.
\end{proof}

Recall that every time of the 2nd or 3rd kind is associated with a pole of $\omega_{0,1}$, and we denote by $p_-\in\sigma(\cP^{\infty}_+)$ the corresponding pole inside $C_-$, and as a consequence $p_+:=\sigma(p_-)$ is inside $C_+$ if it is not a ramification point whereas $p_+=p_-$ is inside $C_-$ if it is a ramification point --- recall that we are not allowing a deformation such that $p_\pm$ approach to each other. Then, we can show the following property:
\begin{lem}\label{lem:varpole}
The following function in $p_0$
\begin{equation}
    \frac{\delta_t^{(1)}*\omega_{0,1}(p_0)}{\omega_{0,1}(p_0)}\label{dpolesw01}
\end{equation}
is holomorphic at $p_0=p_\pm$, and for $2g-2+n\geq-1$, the following differential in $p_0$ is regular at $p_0=p_\pm$:
    \begin{equation}
        \int_{q\in\gamma}\Lambda(q)\cdot\omega_{g,n+2}(q,p_0,J).\label{dpoles}
    \end{equation}
\end{lem}
\begin{proof}
    \eqref{dpolesw01} means that the pole order of $\omega_{0,1}$ does not get higher even after taking the variation. This is because we are only considering generalised cycles $(\gamma,\Lambda)$, which by definition guarantees that the pole order of $\Lambda(q)$ at $q=p_\pm$ is at most $m-1$ for an $m$-th order pole of $\omega_{0,1}$. Then, since $\omega_{0,2}(q,p_0)$ has a double pole at $q=\sigma(p_0)$, we find that $\delta_t^{(1)}*\omega_{0,1}$ has at most an $m$-th order pole, hence \eqref{dpolesw01} is regular as a function in $p_0$ at $p_0=p_\pm$.
    
    As for \eqref{dpoles1}, Theorem \ref{thm:RTR} shows that all poles of $\omega_{g,n+2}(q,p_0,J)$ for $2g-2+n\geq-1$ with respect to $q$ lie in $\cR^*\cup\sigma(J_0\cup\cP_+^0)$. Thus, a pole of \eqref{dpoles} at $p_0=p_+$ can only come from the pole of the integrand $\omega_{g,n+2}(q,p_0,J)$ at $q=\sigma(p_0)$, hence we focus on this contribution.

    As derived in \cite[Section 3.3]{O23}, the pole of $\omega_{\frac12,2}(q,p_0)$ at $q=\sigma(p_0)$ arises from the following term:
    \begin{equation}
        \omega_{\frac12,2}(q,p_0)=-\frac{\sQ}{4}\left(d_q\frac{\Delta_q \omega_{0,2}(q,p_0)}{\omega_{0,1}(q)}+d_{p_0}\frac{\Delta_{p_0} \omega_{0,2}(q,p_0)}{\omega_{0,1}(p_0)}\right)+\text{reg. at $q=\sigma(p_0)$}.\label{dpoles1}
    \end{equation}
  Recall from Definition \ref{def:pair} that $\Lambda(q)$ has a pole at $q=p_+$ of order at most $m-1$ if $p_+$ is an $m$-th  order pole of $\omega_{0,1}$. Then, since there is $\omega_{0,1}$ in the denominator of \eqref{dpoles1}, one notices that \eqref{dpoles} becomes regular at $p_0=p_\pm$ after integration.

    We now proceed by induction in $\chi=2g-2+n\geq-1$. Since $\Sigma=\bP^1$, the integrand of the refined topological recursion formula is a meromorphic differential in $p$. Thus, by using the property that the sum of all residues of a meromorphic differential is zero on a compact Riemann surface, one can  rewrite the recursion formula as
\begin{align}
    \omega_{g,n+2}(p_0,q,J)&=\frac{1}{2\pi i}\int_{p\in C_+}\frac{\eta^p(p_0)}{2\omega_{0,1}(p)}{\rm Rec}_{g,n+2}(p,q,J)\nonumber\\
    &=-\frac{1}{2\omega_{0,1}(p_0)}\cdot{\rm Rec}_{g,n+2}(p_0,q,J)+R_{g,n+2}(p_0,q,J),\label{RLE}
\end{align}
where 
\begin{align}
    R_{g,n+2}(p_0,q,J)=&\frac{1}{2\pi i}\int_{p\in C_+\backslash \{p_0\}}\frac{\eta^p(p_0)}{2\omega_{0,1}(p)}{\rm Rec}_{g,n+1}(p,q,J)\nonumber\\
    =&d_q\left(\frac{\eta^q(p_0)}{2\omega_{0,1}(q)}\cdot\omega_{g,n+1}(q,J)\right)+\frac{1}{2\pi i}\int_{p\in C_+\backslash \{p_0,q\}}\frac{\eta^p(p_0)}{2\omega_{0,1}(p)}{\rm Rec}_{g,n+1}(p,q,J),\label{dpoleR}
\end{align}
and $C_+\backslash \{p_0\}$ denotes the resulting contour after evaluating residue at $p_0=p_\pm$ which gives the first term in the second line of \eqref{RLE}, and similarly for $C_+\backslash \{p_0,q\}$. \eqref{RLE} is indeed called the \emph{refined loop equation of type $(g,n+2)$} \cite{O23}.

Then, by the induction ansatz, we have
\begin{align}
    &\frac{1}{2\omega_{0,1}(p_0)}\cdot\int_{q\in\gamma}\Lambda(q)\cdot{\rm Rec}_{g,n+2}(p_0,q,J)\nonumber\\
    &=\frac{1}{2\omega_{0,1}(p_0)}\cdot \int_{q\in\gamma}\Lambda(q)\cdot2\omega_{0,2}(p_0,q)\cdot\omega_{g,n+1}(p_0,J)+\text{reg at $p_0=p_\pm$}\nonumber\\
    &=\frac{1}{\omega_{0,1}(p_0)}\cdot\left(\delta_t^{(1)}*\omega_{0,1}(p_0)\right)\cdot\omega_{g,n+1}(p_0,J)+\text{reg at $p_0=p_\pm$}.\label{dpoleRec}
\end{align}
Thus, this contribution is non-singular at $p_0=p_\pm$ thanks to the first statement of this lemma. On the other hand, the contribution from $R_{g,n+2}$ can be written as
\begin{align}
    \int_{q\in\gamma}\Lambda(q)\cdot R_{g,n+2}(p_0,q,J)=\int_{q\in\gamma}\Lambda(q)\cdot d_q\left(\frac{\eta^q(p_0)}{2\omega_{0,1}(q)}\cdot\omega_{g,n+1}(q,J)\right)+\text{reg at $p_0=p_\pm$}.
\end{align}
The first term vanishes no matter if $t$ is of the 2nd kind or 3rd kind due to the pole structure of $\Lambda(q)$. Therefore, we conclude that \eqref{dpoles1} is regular at $p_0=p_\pm$.
\end{proof}

\subsubsection{Proof of Theorem \ref{thm:var} for $\omega_{g,n+1}$}

We now prove Theorem \ref{thm:var} for $\omega_{g,n+1}$ by induction in $\chi=2g-2+n\geq-2$. For $\chi=-2$, i.e, $(g,n)=(0,0)$, it holds because we only consider parameters associated with generalised cycles (c.f. \cite{EO07,E17}). For $\chi=-1$, the theorem also holds because it is shown for $\delta_t^{(2)}*\omega_{0,2}(p_0,p_1)$ in \cite{EO07}, and also because we assume that a refined spectral curve $\cS_{\bm\mu}(t)$ satisfies the refined deformation condition. Our approach is similar to the technique shown in \cite{CEM10,BCCGF22} to some extent.

Let us assume that the variational formula holds up to $\chi=k$ for some $k\geq-1$, and we consider the case for $(g,n)$ with $\chi=2g-2+n=k+1$. Then by applying the variational operator to the recursion formula in the form of the first line of \eqref{RLE}, Lemma \ref{lem:dRec} imply
\begin{align}
&\delta_t^{(n+1)}*\omega_{g,n+1}(p_0,J)\nonumber\\
    &=\frac{1}{2\pi i}\int_{p\in C_+}\frac{\eta^p(p_0)}{2\omega_{0,1}(p)}\cdot\delta_t^{(n+1)}*{\rm Rec}_{g,n+1}(p,J)\nonumber\\
    &\;\;\;\;+\frac{1}{2\pi i}\int_{p\in C_+}\frac{\eta^p(p_0)}{2\omega_{0,1}(p)}\cdot2\delta_{\epsilon}^{(1)}*\omega_{0,1}(p)\cdot\left(\frac{1}{2\pi i}\int_{q\in C_+}\frac{\eta^q(p)}{2\omega_{0,1}(q)}{\rm Rec}_{g,n+1}(q,J)\right)\nonumber\\
    &=\frac{1}{2\pi i}\int_{p\in C_+}\int_{q\in\gamma}\Lambda(q)\cdot\frac{\eta_\sA^p(p_0)}{2\omega_{0,1}(p)}{\rm Rec}_{g,n+2}(p,q,J),\label{dw1}
\end{align}
where at the second equality we used the induction ansatz on $\delta_t^{(n+1)}*{\rm Rec}_{g,n+1}(p,J)$ and also we applied the recursion formula in the third line to obtain $\omega_{g,n+1}(p,J)$.

Let us simplify \eqref{dw1}. Consider a decomposition $C_+=C_0\cup C_\gamma$ such that $C_\gamma$ contains $p_+$ inside but no other poles of the integrand. Then, $C_0$ and $\gamma$ do not intersect and one can freely exchange the order of integration. In particular, one obtains:
\begin{align}
    \delta_t^{(n+1)}*\omega_{g,n+1}(p_0,J)-\int_{q\in\gamma}\Lambda(q)\cdot\omega_{g,n+2}(p,q,J)=\rho_{g,n+1}(p_0,J),\label{dw2}
\end{align}
where
\begin{align}
    \rho_{g,n+1}(p_0,J):=\left(\Res_{p=p_+}\int_{q\in\gamma}-\int_{q\in\gamma}\Res_{p=q}\right)\Lambda(q)\cdot\frac{\eta^p(p_0)}{2\omega_{0,1}(p)}{\rm Rec}_{g,n+2}(p,q,J).\label{dw3}
\end{align}
Note that the first term in \eqref{dw3} is the remnant contribution from $C_\gamma$ whereas the second term is the counter effect of applying the refined recursion formula \eqref{RLE} to obtain $\cdot\omega_{g,n+2}(p,q,J)$ on the left-hand side of \eqref{dw2}. As shown in \eqref{dpoleRec} in Lemma \ref{lem:varpole}, the integrand of the first term in \eqref{dw3} as a differential in $p$ becomes regular at $p=p_\pm$, hence it vanishes. Furthermore, since $\omega_{0,2}(p,\sigma(q))$ is the only term that has a pole at $p=q$ in the integrand in \eqref{dw3}, the second term can be written as
\begin{equation}
    \int_{q\in\gamma}\Res_{p=q}\Lambda(q)\cdot\frac{\eta^p(p_0)}{2\omega_{0,1}(p)}{\rm Rec}_{g,n+2}(p,q,J)=\int_{q\in\gamma}\Lambda(q)\cdot d_q\left(\frac{\eta^q(p_0)}{2\omega_{0,1}(q)}\omega_{g,n+1}(q,J)\right).
\end{equation}
This always vanishes for any generalised cycle due to the pole order of $\Lambda(q)$ at $q=p_\pm$ (see Definition \ref{def:pair}). This completes the proof for $\omega_{g,n+1}$.

\subsection{Proof of Theorem \ref{thm:var}: for $F_g$}\label{app:varproofF}
Notice that the above proof was based on the pole structure of the refined topological recursion formula, or equivalently, refined loop equations. Since $F_g$ does not appear in the recursion formula, we need a different approach to prove for $F_g$.\footnote{Strictly speaking, the original proof in \cite{EO07} is based on the rooted-graph interpretation of the Eynard-Orantin recursion formula which works only for $\omega_{g,n+1}$, but not for $F_g$. The statement itself still stands for $F_g$ too as one can easily see in \eqref{dFg1} whose computation is valid beyond hyperelliptic curves. Alternatively, one can simply introduce a non-tooted graphical interpretation for the defining equation of $F_g$ and properly make sense of the action of the variation, which is perhaps just omitted in \cite{EO07}.}.

For $g>1$, we directly take the derivative of the definition of $F_g$ which gives
\begin{align}
    \frac{\partial F_g}{\partial t}=&\frac{1}{2-2g}\frac{1}{2\pi i}\oint_{p\in C^{\mathfrak p}_-}\left(\left(\delta_t^{(1)}*\phi(p)\right)\cdot\omega_{g,1}(p)+\phi(p)\cdot\left(\delta_t^{(1)}*\omega_{g,1}(p)\right)\right)\nonumber\\
    =&\frac{1}{2-2g}\frac{1}{2\pi i}\oint_{p\in C^{\mathfrak p}_-}\int_{q\in\gamma}\Lambda(q)\cdot\left(\left(\int^p\omega_{0,2}(q\cdot)\right)\cdot\omega_{g,1}(p)+\phi(p)\cdot\omega_{g,2}(p,q)\right)\label{dFg0}
\end{align}
where we used Lemma \ref{lem:chain}, and we used the variational formula for $\delta_t^{(1)}*\omega_{g,1}$ at the second equality. Then, since $C_-^{\mathfrak{p}}$ does not contain any point in $\cP^{\infty}$, we can exchange the order of integration with respect to $p$ and $q$ in \eqref{dFg0}. After some manipulation by using the dilaton equation \eqref{dilaton}, we find
\begin{equation}
    \frac{\partial F_g}{\partial t}-\int_{q,\in\gamma}\Lambda(q)\cdot\omega_{g,1}(q)=\frac{1}{2-2g}\int_{q,\in\gamma}\Lambda(q)\cdot\Res_{p=\sigma(q)}\phi(p)\cdot\omega_{g,2}(p,q),\label{dFg1}
\end{equation}
 where the right-hand side is the counter effect of applying the dilaton equation, similar to \eqref{dw3}. Therefore, what we have to show is that the right-hand side of \eqref{dFg1} vanishes. This is straightforward when $\sQ=0$ because $\omega_{g,n+1}(p_0,J)|_{\sQ=0}$ have no poles at $p_0=\sigma(p_i)$. However, since the pole structure is different in the refined setting, the proof involves more careful considerations.

\subsubsection{Proof for the 2nd kind}
We first consider the case where $t$ is a 2nd kind time (Definition \ref{def:pair}). That is, for $m\geq2$, we assume that $\omega_{0,1}$ has a pole at $p_\pm$ of order $m$, $\Lambda(q)$ is meromorphic at $q=p_\pm$ of order $l$ where $l\in\{1,..,m-1\}$, and $\gamma$ is a small contour encircling $p_\pm$ in the prescribed orientation. Therefore, the integral simply reduces to taking residue at $q=p_\pm$, and as a consequence, it is sufficient to check the order of the zero of $\Res_{p=\sigma(q)}\phi(p)\cdot\omega_{g,2}(p,q)$. This is a clear contrast from the 3rd kind cases at which one has to consider open-contour integrals.

Our task is to show the following property which immediately implies the variational formula for the 2nd kind:
\begin{prop}\label{prop:II}
    Let us define a multidifferential $I_{g,n+1}$ as follows:
    \begin{equation}
        I_{g,n+1}(q,J):=\Res_{p=\sigma(q)}\phi(p)\cdot\omega_{g,n+2}(q,J,p).
    \end{equation}
    Then, we have $I_{0,1}=\omega_{0,1}$, $I_{0,2}=0$, $I_{\frac12,1}=0$, and for $2g-2+n\geq0$, it can be written as follows:
    \begin{equation}
        I_{g,n+1}(q,J)=-\frac12\omega_{g,n+1}(q,J)-\frac12\omega_{g,n+1}(\sigma(q),J)-\sQ \cdot d_q\frac{I_{g-\frac12,n+1}(q,J)}{2\omega_{0,1}(q)}+\tilde I_{g,n+1}(q,J),\label{dFg2}
    \end{equation}
    where $\Delta_q\tilde I_{g,n+1}(q,J)$ has at least an $m$-th zero at $q=p_\pm$.
\end{prop}
\begin{proof}
    It is trivial to see that $I_{0,1}=\omega_{0,1}$ and  $I_{0,2}=0$. As discussed in \eqref{dpoles1}, the pole structure of $\omega_{\frac12,2}(p,p_0)$ at $p=\sigma(p_0)$ also immediately implies that $I_{\frac12,1}=0$. For $2g-2+n\geq0$, we proceed by induction and consider refined loop equations \eqref{RLE} for $\omega_{g,n+2}(q,J,p)$ by treating $q$ as the first variable. Let us only give a few useful techniques in order to avoid tedious computational arguments.

    As shown in \eqref{dpoleR} (see also \cite[Proposition 3.17]{O23}), the singular term of $R_{g,n+2}(q,J,p)$ at $p=\sigma(q)$ is written as
    \begin{equation}
        R_{g,n+2}(q,J,p)=d_p\left(\frac{\eta^p(q)}{2\omega_{01,}(p)}\cdot\omega_{g,n+1}(p,J)\right)+\text{ reg. at $p=\sigma(q)$}.
    \end{equation}
    Thus, we have
    \begin{equation}
        \Res_{p=\sigma(q)}\phi(p)\cdot R_{g,n+2}(q,J,p)=-\frac{1}{2}\omega_{g,n+1}(\sigma(q),J).
    \end{equation}
    Notice that the above term is the only contribution from $R_{g,n+2}(q,J,p)$ to $I_{g,n+1}$, which is the second term in \eqref{dFg2}. Therefore, the other terms in \eqref{dFg2} are all coming from ${\rm Rec}_{g,n+2}(q,J,p)$ in \eqref{RLE}.
    
    The first term in \eqref{dFg2} is the contribution of $\omega_{0,2}(q,p)$ in ${\rm Rec}_{g,n+2}(q,J,p)$, more explicitly,
    \begin{equation}
         \Res_{p=\sigma(q)}\phi(p)\cdot\left(-\frac{1}{2\omega_{0,1}(q)}\left(2\omega_{0,2}(q,p)+\frac{dx(q)dx(p)}{(x(q)-x(p))^2}\right)\cdot\omega_{g,n+1}(q,J)\right)=-\frac{1}{2}\omega_{g,n+1}(q,J).
    \end{equation}
    Next, terms involving $\omega_{\frac12,1}$ in ${\rm Rec}_{g,n+2}(q,J,p)$ give
    \begin{align}
        &\Res_{p=\sigma(q)}\phi(p)\cdot\left(-\frac{1}{2\omega_{0,1}(q)}\left(2\omega_{\frac12,1}(q)\cdot\omega_{g-\frac12,n+2}(q,J,p)+\sQ\cdot dx(q)\cdot d_q\frac{\omega_{g-\frac12,n+2}(q,J,p)}{dx(q)}\right)\right)\nonumber\\
        &=-\frac{\Delta\omega_{\frac12,1}(q)}{2\omega_{0,1}(q)}\cdot I_{g-\frac12,n+1}(q,J)-\sQ \cdot d_q\frac{I_{g-\frac12,n+1}(q,J)}{2\omega_{0,1}(q)}.\label{dFg3}
    \end{align}
The last term in \eqref{dFg3} coincides with the third term in \eqref{dFg2}. Note that the first term in \eqref{dFg3} only has an $(m-1)$-order zero at $q=p_\pm$ due to the presence of $\omega_{\frac12,1}(q)$, but
\begin{equation}
    \frac{\Delta\omega_{\frac12,1}(q)}{2\omega_{0,1}(q)}\cdot \Delta_q I_{g-\frac12,n+1}(q,J)
\end{equation}
has a higher order zeroe thanks to the induction ansatz. Then, one can easily see that all other terms have the prescribed zero behaviour thanks to the $\omega_{0,1}(q)$ in the denominator in the refined loop equation \eqref{RLE}. 
\end{proof}

\subsubsection{Proof for the 3rd kind}
We will show an analogous proposition to Proposition \ref{prop:II} but in a slightly different form. First recall that $\omega_{g,n+1}(p_0,J)$ for $2g-2+n\geq0$ has no residue with respect to $p_0$. Thus, the following residue makes sense:
\begin{equation}
    I_{g,n+1}^{*}(q,J):=\left(\Res_{p=p_+}+\Res_{p=p_-}\right)\omega_{0,1}(p)\cdot\int^{p}_{\sigma(p)}\omega_{g,n+2}(q,J,\cdot),\label{I*}
\end{equation}
where the integral is taken with respect to the last variable.
\begin{lem}\label{lem:II}
    $I^*_{0,1}(q)/\omega_{0,1}(q)$ and $I_{g,n+1}(q,J)$ for $2g-2+n\geq-1$ are regular at $q=p_\pm$.
\end{lem}
\begin{proof}
    For $I^*_{0,1}(q)$, we have
    \begin{equation}
        I^*_{0,1}(q)=\left(\Res_{p=p_+}+\Res_{p=p_-}\right)\omega_{0,1}(p)\cdot\eta^p(q).
    \end{equation}
    Thus, $I^*_{0,1}$ picks up the singular part of $\omega_{0,1}$ at $q=p_+$ and $q=p_-$ (c.f. \cite[Section 2]{O23}). Thus, it becomes regular after dividing by $\omega_{0,1}(q)$.
    
    For $I_{g,n+1}(q,J)$ for $2g-2+n\geq-1$, since the proposition only concerns a local behaviour at $q=p_\pm$, potentially singular terms may appear only from the pole of $\omega_{g,n+2}(q,J,p)$ at $p=\sigma(q)$ and we only focus on these poles similar to above discussions. Then, for the rest of the proof we apply the same technique as the proof of Lemma \ref{lem:varpole} and Proposition \ref{prop:II}. That is, we treat the contributions from $\omega_{0,2}$ and $\omega_{\frac12,1}$ differently, and check the singular behaviour at $q=p_\pm$ by induction. Since arguments will be almost parallel to the one given in Lemma \ref{lem:varpole} and Proposition \ref{prop:II}, we omit it.
\end{proof}

We now prove the variational formula for the 3rd kind. Lemma \ref{lem:II} implies that $I^*_{g,n+1}(q,J)$ as a differential in $q$ has no residue everywhere on $\Sigma$. This is because $\omega_{g,n+2}(q,J,p)$ has no residue with respect to $q$ (Theorem \ref{thm:RTR}), and thus residues can only potentially appear at $q=p_\pm$ after taking the integral \eqref{I*} which we have just shown that this is not the case. Thus, we can consider integration once more:
\begin{align}
    I^{**}_{g,n}(J):=&\left(\Res_{q=p_+}+\Res_{q=p_-}\right)\omega_{0,1}(q)\cdot \int^q_{\sigma(q)}I^*_{g,n+1}(\cdot,J)\nonumber\\
    =&\left(\Res_{q=p_+}+\Res_{q=p_-}\right)\left(\Res_{p=p_+}+\Res_{p=p_-}\right)\omega_{0,1}(q)\cdot\omega_{0,1}(p)\cdot\int^{q}_{\sigma(q)}\int^{p}_{\sigma(p)}\omega_{g,n+2}(\cdot,\cdot,J).\label{I**}
\end{align}
Since $\omega_{g,n+2}$ is  symmetric multidifferential, one can simply relabel $p\leftrightarrow q$ in \eqref{I**}. On the other hand, as discussed in \cite[Appendix A]{O23}, exchanging the order of residues would give
\begin{equation}
    \left(\Res_{p=p_+}+\Res_{p=p_-}\right)\left(\Res_{q=p_+}+\Res_{q=p_-}\right)=\left(\Res_{q=p_+}+\Res_{q=p_-}\right)\left(\Res_{p=p_+}+\Res_{p=p_-}+\Res_{p=q}+\Res_{p=\sigma(q)}\right).
\end{equation}
Therefore, we find
\begin{equation}
    \left(\Res_{q=p_+}+\Res_{q=p_-}\right)\left(\Res_{p=q}+\Res_{p=\sigma(q)}\right)\omega_{0,1}(q)\cdot\omega_{0,1}(p)\cdot\int^{q}_{\sigma(q)}\int^{p}_{\sigma(p)}\omega_{g,n+2}(\cdot,\cdot,J)=0.\label{I**1}
\end{equation}

Now notice that 
\begin{align}
    \left(\Res_{p=q}+\Res_{p=\sigma(q)}\right)\omega_{0,1}(p)\cdot \int^{p}_{\sigma(p)}\omega_{g,n+2}(\cdot,q,J)=&2\Res_{p=\sigma(q)}\phi(p)\cdot \omega_{g,n+2}(p,q,J)\nonumber\\
    =&2I_{g,n+1}(q,J).
\end{align}
Thus, with the help of Proposition \ref{prop:II}, the left-hand side of \eqref{I**1} can be written as
\begin{align}
    \text{L.H.S. of \eqref{I**1}}=&2\left(\Res_{q=p_+}+\Res_{q=p_-}\right)\omega_{0,1}(q)\cdot\int^{q}_{\sigma(q)}I_{g,n+1}(\cdot,J)\nonumber\\
    =&4t\int^{p_+}_{p_-}I_{g,n+1}(\cdot,J),\label{I**2}
\end{align}
where $t$ is the 3rd kind time at $p_\pm$. Note that there is no contribution from higher order poles of $\omega_{0,1}(q)$ thanks to Proposition \ref{prop:II}. Combining \eqref{I**1} and \eqref{I**2}, we conclude that the right-hand side of \eqref{dFg1} vanishes for the 3rd kind as well.

This completes the proof of Theorem \ref{thm:var} for both $\omega_{g,n+1}$ and $F_g$.

\subsection{Explicit computations}\label{app:rdc}
We will provide explicit computational results for Proposition \ref{prop:rdc1}, \ref{prop:rdc2}, and \ref{prop:QQC}.

\subsubsection{Proof of Proposition \ref{prop:rdc1}}\label{sec:proofrdc1}
We will give computations for the Weber, Whittaker, and Bessel curve as evidence of Proposition \ref{prop:rdc1}. The statement for other curves can be similarly checked. Throughout Section \ref{sec:proofrdc1}, we let $z$ be a coordinate on $\Sigma=\bP^1$, and we parameterise $y(z),x(z)$ such that $\cR=\{1,-1\}$ and $\cP=\{0,\infty\}$, some of which are different from rational expressions in \cite{IKT18-2} but they are related by an appropriate M\"{o}bius transformation. As shown in \cite{IKT18-2}, the corresponding generalised cycle $(\Lambda,\gamma)$ is given such that $\Lambda=1$ and $\gamma$ is a contour from $0$ to $\infty$. We denote by $t$ the corresponding 3rd kind time, which is none other than $\lambda$ in \cite{IKT18-2}. Furthermore, we choose $\cP_+=\{\infty\}$ to define a refined spectral curve and denote by $\mu$ for the associated complex parameter. 

 \begin{description}
 \item[Weber] The underlying curve is given by
 \begin{equation}
     y^2-\frac{x^2}{4}+t=0,\quad y(z)=\frac{\sqrt{t}}{2}\left(z-\frac{1}{z}\right),\quad x(z)=\sqrt{t}\left(z+\frac{1}{z}\right).
 \end{equation}
Then, the variational operator and $\omega_{\frac12,1}$ are respectively given as
\begin{equation}
    \delta_t^{(1)}=\frac{\partial}{\partial t} -\frac{z \left(z^2+1\right)}{2 (z-1) (z+1) t}\frac{\partial}{\partial z}
\end{equation}
\begin{equation}
    \omega_{\frac12,1}(z_0)=\frac{\sQ}{2}\left(-\frac{z_0^2+1}{(z_0-1) z_0 (z_0+1)}-\frac{\mu}{z_0}\right)dz_0.
\end{equation}
$\omega_{\frac12,2}$ itself is lengthy to write down here, but we have
\begin{equation}
    \delta_t^{(1)}*\omega_{\frac12,1}(z_0)=\int_0^{\infty}\omega_{\frac12,2}(\cdot,z_0)= -\sQ\frac{z_0 \left(-\mu +\mu  z_0^2+2 z_0^2+2\right)}{(z_0-1)^3 (z_0+1)^3 t}dz_0.
\end{equation}

     \item[Whittaker] The underlying curve is given by
 \begin{equation}
     xy^2-\frac{x}{4}+t=0,\quad y(z)=\frac{z-1}{2 (z+1)},\quad x(z)=\frac{t (z+1)^2}{z}.
 \end{equation}
Then, the variational operator and $\omega_{\frac12,1}$ are respectively given as
\begin{equation}
    \delta_t^{(1)}=\frac{\partial}{\partial t} -\frac{z (z+1)}{t (z-1)}\frac{\partial}{\partial z}
\end{equation}
\begin{equation}
    \omega_{\frac12,1}(z_0)=\frac{\sQ}{2}\left(\frac{1}{1-z_0^2}-\frac{\mu}{z_0}\right)dz_0.
\end{equation}
$\omega_{\frac12,2}$ itself is lengthy to write down here, but we have
\begin{equation}
    \delta_t^{(1)}*\omega_{\frac12,1}(z_0)=\int_0^{\infty}\omega_{\frac12,2}(\cdot,z_0)= -\sQ\frac{-\mu +\mu  z_0+z_0+1}{t \left(z_0-1\right)^3}dz_0.
\end{equation}

\item[Bessel] The underlying curve is given by
 \begin{equation}
     x^2y^2-\frac{x}{4}-t^2=0,\quad y(z)=-\frac{z^2-1}{16 t z},\quad x(z)=-\frac{16 t^2 z}{(z+1)^2}.
 \end{equation}
Then, the variational operator and $\omega_{\frac12,1}$ are respectively given as
\begin{equation}
    \delta_t^{(1)}=\frac{\partial}{\partial t} +\frac{2 z (z+1)}{t (z-1)}\frac{\partial}{\partial z}
\end{equation}
\begin{equation}
    \omega_{\frac12,1}(z_0)=\frac{\sQ}{2}\left(\frac{z_0^2+1}{z_0(1-z_0^2)}-\frac{\mu}{z_0}\right)dz_0.
\end{equation}
$\omega_{\frac12,2}$ itself is lengthy to write down here, but we have
\begin{equation}
    \delta_t^{(1)}*\omega_{\frac12,1}(z_0)=\int_0^{\infty}\omega_{\frac12,2}(\cdot,z_0)= -2\sQ\frac{-\mu +\mu  z_0+z_0+1}{t \left(z_0-1\right)^3}dz_0.
\end{equation}
     
 \end{description}

\subsubsection{Proof of Proposition \ref{prop:rdc2}}
Recall that the parametrisation of the curve is given in \eqref{rational} as
\begin{equation}
    x(z)=z^2-2q_0,\quad y(z)=2z(z^2-3q_0)=2z(z^2-q_z^2),\quad q_0=\sqrt{-\frac{t}{6}},\quad q_z:=\sqrt{3q_0}.
\end{equation}
Then, we can explicitly construct the variational operator and $\omega_{\frac12,1}$ as
\begin{equation}
    \delta_t^{(1)}=\frac{\partial}{\partial t} -\frac{1}{2z\sqrt{6t}}\frac{\partial}{\partial z}
\end{equation}
\begin{equation}
    \omega_{\frac12,1}(z_0)=\sQ\left(-\frac{1}{2 z_0}+\frac{\mu -1}{2 \left(z_0-q_z\right)}+\frac{-\mu -1}{2 \left(q_z+z_0\right)}\right)dz_0.
\end{equation}
$\omega_{\frac12,2}$ itself is complicated to write down here, but we have

\begin{align}
  \delta_t^{(1)}\omega_{\frac12,1}(z_0)=& \sQ\bigg(-\frac{2 (\mu +2) z_0 q_z^2+2 (2 \mu +1) z_0^2 q_z+2 q_z^3+(2 \mu +1) z_0^3}{8 z_0^3 q_z^3 \left(q_z+z_0\right){}^2}\nonumber\\
  &\qquad+\frac{(\mu -1) \left(q_z^2+z_0^2\right)}{8 q_z^3 \left(q_z^2-z_0^2\right)^2}\bigg)dz_0,\label{dw1/2,1}\\
  \Res_{z=\infty}\Lambda_t(z)\cdot\omega_{\frac12,2}(z,z_0)=&\sQ\left(-\frac{2 (\mu +2) z_0 q_z^2+2 (2 \mu +1) z_0^2 q_z+2 q_z^3+(2 \mu +1) z_0^3}{8 z_0^3 q_z^3 \left(q_z+z_0\right)^2}\right)dz_0.
\end{align}
Therefore, they become the same if and only if $\mu=1$. It is worth noting that the second term in \eqref{dw1/2,1} is singular at $z_0=-q_z$ whereas $\omega_{\frac12,2}$ will never have a pole at $\pm q_z$ due to Theorem \ref{thm:RTR} and Lemma \ref{lem:varpole}. This is a clear contrast from hypergeometric type curves that $\omega_{\frac12,1}$ has residue at $\sigma(\cP_+)$, but its variation $\delta_t*\omega_{\frac12,1}$ is regular as shown in Section \ref{sec:proofrdc1}. It is interesting to investigate whether this phenomenon arises from the difference between $\cP^0$ and $\cP^\infty$, or difference between 2nd kind and 3rd kind.

\subsubsection{Proof of Proposition \ref{prop:QQC}}
Without explicit computation, it is not hard to see from the definition of $R_{1,1}^{\sQ\text{-{\rm top}}}(z)$ that it has a triple pole at $p=p_\pm$ whose coefficient is proportional to $(\mu-1)(\mu-3)$. This can be checked by looking at the contribution of $\omega_{\frac12,1}$ whose pole structure is given in \eqref{mu}. What is less straightforward without explicit computation is to show that the subleading order coefficient is still proportional to $(\mu-1)$ but not to $(\mu-3)$ anymore. By explicit computations, we have
\begin{align}
    R_{1,1}^{\sQ\text{-{\rm top}}}(z)=&\bigg(\frac{(1-\mu) \left(6 \mu  z^2 q_z^2-18 z^2 q_z^2+9 \mu  q_z^4-11 q_z^4-7 \mu  z^4+5 z^4\right)}{64 q_z^4
   \left(q_z-z\right){}^3 \left(q_z+z\right){}^3}\nonumber\\
    &\quad-\frac{15 \mu ^2+7}{64 q_z^4 \left(q_z-z\right) \left(q_z+z\right)}\bigg)dz.\label{R11top}
\end{align}
This clearly shows that the $\sQ$-top quantisation condition is satisfied only if $\mu=1$. Once we set $\mu=1$, then ${\rm Rec}_{g,1}^{\sQ\text{-{\rm top}}}(z)$ is regular at $z=q_z$ for all $g>1$ so is $R_{g,1}^{\sQ\text{-{\rm top}}}(z)$ because $\omega_{\frac12,1}(z)$ is regular at $z=q_z$. This completes the proof of Proposition \ref{prop:QQC}.

\subsection{Proof of Theorem \ref{thm:QC}}\label{app:proofQC}
We will consider the contribution of $F_g$ for each $g\geq0$.

\subsubsection{$F_{0}$}
It is already shown in \cite{IS15} that 
\begin{equation}
    F_0=-\frac{48}{5}q_0^5,
\end{equation}
from which we find
\begin{equation}
    \frac{\partial F_0}{\partial t}=4q_0^3.
\end{equation}
This is consistent with the $\epsilon^0_1$ term in Theorem \ref{thm:Qtop}, and also consistent with the unrefined quantum curve in \cite{IS15}. 

\subsubsection{$F_{\frac12}$}
We take the definition of $F_{\frac12}$ as in Theorem \ref{thm:QC}, which gives
\begin{equation}
    \frac{\partial^2}{\partial t^2}F_{\frac12}^{\sQ\text{-{\rm top}}}:=\Res_{z_0=0}\Res_{z_1=0}\cdot \Lambda_t(z_0)\cdot \Lambda_t(z_1)\cdot\varpi_{\frac12,2}(z_0,z_1)=\frac{1}{4}\left(-\frac32\right)^{\frac14}t^{\frac14}.\label{dFdt2}
\end{equation}
Then from the boundary condition set in Theorem \ref{thm:QC}, one finds that
\begin{equation}
     F_{\frac12}^{\sQ\text{-{\rm top}}}=\frac{4}{5}\left(-\frac32\right)^{\frac14}t^{\frac54},\quad Q_1^{\sQ\text{-{\rm top}}}(z_0)\big|_{\mu=1}=2\frac{\partial  F_{\frac12}^{\sQ\text{-{\rm top}}}}{\partial t}.
\end{equation}

One may wonder why we do not consider a solution $\tilde F_{\frac12}^{\sQ\text{-{\rm top}}}$ by respecting the variational formula for $\omega_{\frac12,1}$ as:
\begin{equation}
    \frac{\partial \tilde F_{\frac12}}{\partial t}:=\Res_{z=\infty}\Lambda_t(z)\cdot\varpi_{\frac12,1}(z)=\mu\cdot q_z,\label{F1/2,1}
\end{equation}
with unfixed $\mu$. In fact, the condition $\tilde F_{\frac12}(0)=0$ implies that
\begin{equation}
    \tilde F_{\frac12}^{\sQ\text{-{\rm top}}}=\frac{4}{5}\left(-\frac32\right)^{\frac14}t^{\frac54}\cdot\mu,\qquad Q_1^{\sQ\text{-{\rm top}}}(z_0)=2\frac{\partial \tilde F_{\frac12}}{\partial t},
\end{equation}
for any value of $\mu$. One of the issues of taking this definition is that it does not work for the 3rd kind times, because $\omega_{\frac12,1}$ has a pole at the end points of the associated path $\gamma$. Another problem is that if we take \eqref{F1/2,1} as the defining equation for a general spectral curve $\cS_{\bm\kappa,\bm\mu}(\bm t)$, then one cannot show that 
\begin{equation}
    \frac{\partial^2 F_{\frac12}^{\sQ\text{-{\rm top}}}}{\partial t_k\partial t_l}=\oint_{p_0\in\gamma_l}\delta_{t_l}^{(1)}*\left(\Lambda_k(p_0)\cdot\varpi_{\frac12,1}(p_0)\right)
\end{equation}
is symmetric in $k\leftrightarrow l$ or not.

\begin{rem}\label{rem:F1/2}
    The above observation may motivate one to propose a definition of $F_\frac12$ as:
    \begin{equation}
        \forall k,l\in\{1,..,|\bm t|\}\qquad\frac{\partial^2 F_\frac12}{\partial t_k\partial t_l}:=\int_{p\in\gamma_k}\int_{q\in\gamma_l}\Lambda_k(p)\cdot\Lambda_l(q)\cdot\omega_{\frac12,2}(p,q).\label{defF1/2}
    \end{equation}
    for a general refined spectral curve $\cS_{\bm\kappa,\bm\mu}(\bm t)$ satisfying the refined deformation condition. The above definition makes sense, that is, one can show that it is symmetric under $k\leftrightarrow l$ by utilising \eqref{dpoles1} and the anti-invariantness of generalised cycles under the involution $\sigma$. Indeed, this definition works for all hypergeoemtric curves as well as this example. Therefore, we propose that $F_\frac12$ is defined as \eqref{defF1/2}, which is defined uniquely up to a constant and linear dependence in $\bm t$. Similarly, up to constant, we propose that $F_1$ for a refined spectral curve satisfying the refined deformation condition is defined by
    \begin{equation}
        \forall k\in\{1,..,|\bm t|\}\qquad\frac{\partial F_1}{\partial t_k}:=\int_{p\in\gamma_k}\Lambda_k(p)\cdot\omega_{1,1}(p).\label{defF1}
    \end{equation}
    \end{rem}

\subsubsection{Proof of Theorem \ref{thm:QC}.}
We will show that
\begin{equation}
    g\geq1\qquad\frac{\partial F_g}{\partial t}=\Res_{z=\infty}\Lambda_t(z)\cdot\omega_{g,1}(z)=\frac12Q_{2g\geq2}^{\sQ\text{-{\rm top}}}(z_0),\label{QCF}
\end{equation}
where the first equality is merely the variational formula. The second equality means that $Q_{2g\geq2}^{\sQ\text{-{\rm top}}}$ is indeed a constant which we will show below.

The proof is similar to that in \cite{IS15}. For a refined spectral curve $\cS_{\mu=1}(t)$ satisfying the $\sQ$-top quantisation condition, $\omega_{\frac12,1}$ and ${\rm Rec}_{g,1}^{\sQ\text{-{\rm top}}}(z)$ are regular at $z=q_z$, hence \eqref{Qk>1} implies that there should exist a function $R_g(t)$ such that
\begin{equation}
    R_{g,1}^{\sQ\text{-{\rm top}}}(z_0)=R_g(t)\cdot\frac{dz_0}{z_0^2-q_z^2}.
\end{equation}
Thus, by using the explicit rational expression of $x(z)$ and $y(z)$ given in \eqref{rational}, we have
\begin{equation}
    Q_{2g}^{\sQ\text{-{\rm top}}}(z)=2R_g(t).
\end{equation}
Finally, since $\omega_{0,1}(z_0)$ has a 5-th order pole at $z_0=\infty$, the $\sQ$-top loop equation (the $\sQ$-top degree part of the refined loop equation \eqref{RLE}) implies that 
\begin{equation}
    \Res_{z=\infty}\Lambda_t(z)\cdot\omega_{g,1}(z)=\Res_{z=\infty}\Lambda_t(z)\cdot R_{g,1}(z)=R_g(t).
\end{equation}
Note that from \eqref{R11top}, one finds that
\begin{equation}
    \frac{\partial F_1}{\partial t}=-\frac{11}{48t},\quad F_1=-\frac{11}{48}\log t
\end{equation}
Therefore, \eqref{QCF} holds, and this completes the proof.

\newpage

\printbibliography

\end{document}